\documentclass{amsart}

\usepackage{amssymb}

\usepackage{graphicx}

\usepackage{}

\copyrightinfo{2013}{American Mathematical Society}

\newtheorem{theorem}{Theorem}[section]
\newtheorem{lemma}[theorem]{Lemma}

\theoremstyle{definition}
\newtheorem{definition}[theorem]{Definition}

\newtheorem{proposition}[theorem]{Proposition}
\newtheorem{cor}[theorem]{Corollary}

\theoremstyle{remark}
\newtheorem{remark}[theorem]{Remark}

\numberwithin{equation}{section}

\newcommand{\op}{\mbox{Op}}
\newcommand{\heis}{{\bf H}}
\newcommand{\jacob}{{\mathcal J}(\Sigma_g)}
\newcommand{\spacetheta}{{\bf \Theta}_N^\Pi}
\newcommand{\sign}{\mbox{sign}}
\newcommand{\hh}{h}

\begin{document}

\title[From theta functions to TQFT]{From classical
theta functions to topological quantum field theory}


\author{R{\u{a}}zvan Gelca}
\address{Department of Mathematics and Statistics,
Texas Tech University, Lubbock, TX 79409}
\email{rgelca@gmail.com}
\thanks{Research of the first author supported by
the NSF, award No. DMS 0604694}

\author{Alejandro Uribe}
\address{Department of Mathematics,
University of Michigan, Ann Arbor, MI 48109}
\email{uribe@math.lsa.umich.edu}
\thanks{Research of the second author supported by the NSF, award No. DMS
 0805878}

\date{}

\begin{abstract}
Abelian Chern-Simons theory relates classical theta functions
to the topological quantum field theory of
the linking number of knots. In this paper we explain how
to derive the constructs of abelian Chern-Simons theory
directly from the theory of classical theta functions.
It turns out that
the theory of theta functions, from the representation
theoretic point of view of A. Weil, is just an instance of
Chern-Simons theory. The group algebra of the
 finite Heisenberg group is described
as an algebra of curves on a surface, and its Schr\"{o}dinger
representation  is obtained as an action on  curves in a
handlebody.  A careful analysis of the discrete
Fourier transform yields the Murakami-Ohtsuki-Okada formula  for
invariants of 3-dimensional manifolds. In this context, we
give an explanation of why the composition of discrete
Fourier transforms
and the non-additivity of the signature of 4-dimensional
manifolds under gluings obey the same formula.
\end{abstract}

\maketitle

\section{Introduction}\label{sec:1}

In this paper we construct the abelian Chern-Simons topological quantum
field theory directly from the theory of classical theta functions, without
the insights of quantum field theory.

It has been known for years, within abelian Chern-Simons
theory,  that classical theta functions
relate to low dimensional topology  \cite{atiyah1}, \cite{witten}.
Abelian Chern-Simons theory is considerably simpler than its
non-abelian counterparts, and has been studied thoroughly (see e.g. \cite{manoliu1},
 \cite{manoliu2}).
Here we do not start with abelian Chern-Simons
theory, but instead  give a direct construction of the associated topological
quantum field theory based on
the theory of  theta functions, and arrive at skein modules
from representation theoretical considerations.

We consider
theta functions in the representation theoretic point
of view of Andr\'e Weil \cite{weil}. As such, the
space of theta functions is endowed with an action of a finite Heisenberg
group (the Schr\"{o}dinger representation), which induces, via a
Stone-von Neumann theorem, the Hermite-Jacobi action of the modular group.
All this structure is what we shall mean by the theory of theta functions.

We show how the group algebra of the  finite Heisenberg group and its Schr\"{o}dinger
representation on the space of theta functions, lead to
algebras of curves on surfaces and their actions on  spaces of curves
in  handlebodies. These notions are formalized using skein modules.

The Hermite-Jacobi representation of the modular group on theta functions
is a discrete analogue
of the metaplectic representation. The modular group  acts by
automorphisms that can be interpreted as discrete Fourier transforms.
We  show that these discrete Fourier transforms can be expressed as
linear combinations of curves. A careful analysis of their structure
and of their relationship to the Schr\"{o}dinger representation
yields the Murakami-Ohtsuki-Okada formula \cite{moo}
of invariants of 3-manifolds.

As a corollary of our discussion we obtain an explanation of why the
composition of discrete Fourier transforms and the non-additivity of the
signature of 4-dimensional manifolds obey the same formula.

The paper uses results and terminology from the theory of theta functions,
quantum mechanics, and low dimensional topology. To make it accessible
to different audiences we include a fair amount of detail.  A more
detailed discussion of these ideas can be found in \cite{gelca}.

Section 2  reviews the theory of  theta functions
on the Jacobian variety of a surface.
The action of the finite Heisenberg group on theta functions is
 defined via Weyl quantization of the Jacobian variety
in a K\"{a}hler polarization. In fact it has
been found recently that Chern-Simons theory is related
to Weyl quantization \cite{gelcauribe1}, \cite{andersen}, and this
was the starting point of our paper.  
The next section exhibits the representation theoretical model
for theta functions. 
In Section~4 we show that this  model for
 theta functions
is topological in nature, and reformulate it using
algebras of curves on surfaces,
together with their action on skeins of curves in handlebodies which are associated to the linking number.

In Section~5 we derive a formula for the  discrete Fourier transform as
a skein. This formula is interpreted in terms of surgery in the
cylinder over the surface.  Section~6 analizes
the exact Egorov identity which relates the Hermite-Jacobi action to
the Schr\"{o}dinger representation. This analysis shows that the topological operation of
handle slides
is allowed over the skeins that represent discrete Fourier transforms, and
this yields in the next section  the abelian Chern-Simons
invariants of 3-manifolds defined by Murakami, Ohtsuki, and Okada. 
 We point out that the
above-mentioned formula was introduced in an ad-hoc manner by its authors
\cite{moo}, our paper  derives it naturally.

Section~8 shows how to associate to
the discrete Fourier transform a 4-dimensional manifold, and
 explains why the
cocycle of the Hermite-Jacobi action is related to that
governing the non-additivity of
the signature of 4-manifolds  \cite{wall}.
Section~9 should be taken as a conclusion; it puts everything
in the context of Chern-Simons theory.

\section{Theta functions}\label{sec:2}
We start with a closed genus $g$ Riemann surface  $\Sigma_g$, and
consider a canonical basis $a_1,a_2,\ldots, a_g,b_1,b_2,\ldots, b_g$
of $H_1(\Sigma_g, {\mathbb R})$, like the one
  in  Figure~\ref{homology}. To it we associate
a basis in the space of holomorphic differential $1$-forms
$\zeta_1,\zeta_2,\ldots, \zeta_g$, defined by the conditions $\int_{a_k}\zeta_j=
\delta_{jk}$, $j,k=1,2,\ldots, g$. The matrix $\Pi$ with entries
$\pi_{jk}=\int_{b_k}\zeta_j$, $ j,k=1,\ldots, g,$
is symmetric with positive definite imaginary part. This means that if
$\Pi=X+iY$, then $X=X^T$, $Y=Y^T$ and $Y>0$.
The  $g\times 2g$ matrix $(I_g,\Pi)$ is called the period matrix of $\Sigma_g$,
its columns $\lambda_1,\lambda_2, \ldots, \lambda_{2g}$,
called periods, generate a
lattice $L(\Sigma_g)$ in
${\mathbb C}^g={\mathbb R}^{2g}$. The complex torus
\begin{eqnarray*}
\jacob ={\mathbb C}^g/L(\Sigma_g)
\end{eqnarray*}
is  the {\em Jacobian variety} of $\Sigma_g$.
 The map
\begin{eqnarray*}
\sum_j\alpha_ja_j+\sum_j\beta_jb_j\mapsto (\alpha_1,\ldots, \alpha_n,\beta_1, \ldots, \beta_n)
\end{eqnarray*}
induces a homeomorphism $H_1(\Sigma_g,{\mathbb R})/H_1(\Sigma_g,
{\mathbb Z})\rightarrow \jacob$.

\begin{figure}[ht]
\centering
\scalebox{.22}{\includegraphics{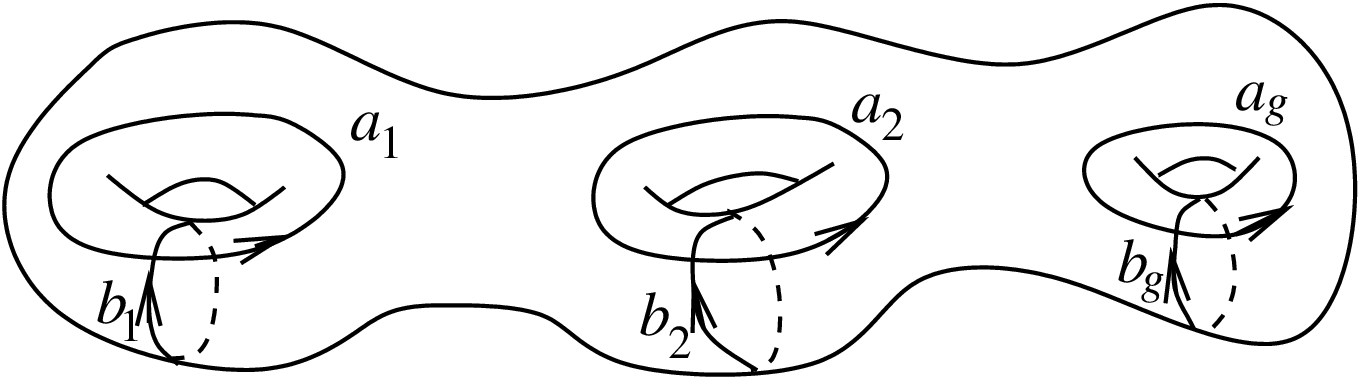}}
\caption{}
\label{homology}
\end{figure}

The complex coordinates $z=(z_1,z_2,\ldots, z_g)$ on $\jacob$
are inherited from ${\mathbb C}^g$. We introduce  real
coordinates $(x,y)=(x_1,x_2,\ldots,x_g,y_1,y_2,\ldots, y_g)$ by
imposing $z=x+\Pi y$. A fundamental domain for the period lattice
in terms of the $(x,y)$ coordinates is
$\{(x,y)\in [0,1]^{2g}\}$.
 $\jacob$  has the canonical
symplectic form 
\begin{eqnarray*}
\omega=2\pi\sum_{j=1}^gdx_j\wedge dy_j.
\end{eqnarray*}
$\jacob$ with the complex structure and  symplectic form
is a K\"{a}hler manifold.
The symplectic form  induces a Poisson bracket on $C^\infty(\jacob)$,
given by $\{f,g\}=\omega(X_f,X_g)$,
where $X_f$ is the Hamiltonian vector field defined by
$df(\cdot)=\omega(X_f,\cdot)$.

{\em Classical theta functions} arise when quantizing $\jacob$
in a K\"{a}hler polarization in the direction
of this Poisson bracket. In this paper
we perform  the quantization in the case where
Planck's constant is the reciprocal of an {\em even} positive
integer: $\hh=\frac{1}{N}$ where $N=2r$, $r\in {\mathbb N}$.
The Hilbert space of the quantization consists of the holomorphic
sections of a line bundle obtained as the tensor product of a line
bundle with curvature $N\omega$ and the square root of the canonical
line bundle. The latter is trivial for the complex torus
and we ignore it. The line bundle with curvature $N\omega$ is the
tensor product of a flat line bundle and the line bundle defined by
the cocycle $\Lambda:{\mathbb C}^g\times L(\Sigma_g)
\rightarrow {\mathbb C}^*$,
\begin{eqnarray*}
\Lambda(z,\lambda_j)=1, \quad \Lambda(z,\lambda_{g+j})=e^{-2\pi i Nz_j-\pi i N\pi_{jj}},
\end{eqnarray*}
$j=1,2,\ldots, g$.
(See e.g.~ \S 4.1.2 of \cite{bgpu} for a discussion of how this cocycle gives rise to a line bundle with
curvature $N\omega$.)
We choose the trivial
flat bundle to tensor with. Then the Hilbert space can
be identified with the space of entire functions on ${\mathbb C}^g$
satisfying the periodicity conditions
\begin{eqnarray*}
f(z+\lambda_j)=f(z), \quad
f(z+\lambda_{g+j})=e^{-2\pi iNz_j-\pi i N\pi_{jj}}f(z).
\end{eqnarray*}
We denote this space by $\spacetheta(\Sigma_g)$; its elements are called
classical theta functions.\footnote{The precise terminology
is canonical theta functions, classical theta functions being
defined by a slight alteration of the periodicity condition. We
use the name classical theta functions
 to emphasize the distinction with the non-abelian theta functions.} A basis of
$\spacetheta(\Sigma_g)$
consists of  the theta series
\begin{eqnarray*}
\theta_\mu^{\Pi}(z)=\sum_{n\in {\mathbb Z}^g}e^{2\pi i N[
\frac{1}{2}\left(\frac{\mu}{N}+n\right)^T \Pi \left(\frac{\mu}{N}+n\right)+
\left(\frac{\mu}{N}+n\right)^Tz]}, \quad \mu\in \{0,1\ldots, N-1\}^g.
\end{eqnarray*}
The  definition of theta series will be extended for convenience to all $\mu\in {\mathbb Z}^g$,
by $\theta_{\mu+N\mu'}=\theta_\mu$ for any $\mu'\in {\mathbb Z}^g$.
Hence the index $\mu$ is taken in ${\mathbb Z}_N^g$.

The inner product 
that makes the theta series into an orthonormal basis is
\begin{eqnarray}\label{innerproduct}
\left<f,g\right>=\left(2 N\right)^{g/2}\det(Y)^{1/2}\int_{[0,1]^{2g}}
f(x,y)\overline{g(x,y)}
e^{-2\pi N y^TYy}dxdy .
\end{eqnarray}
That the theta series form an orthonormal basis is a corollary of the proof of
Proposition~\ref{weylquantization} below.

To define the operators of the quantization, we use the Weyl quantization
method. This quantization method can be defined only on complex
vector spaces, the Jacobian variety is the quotient of such a
space by a discrete group, and the quantization method goes through.
As such, the operator $\op(f)$ associated to a function $f$
 on $\jacob$ is the Toeplitz operator with symbol $e^{\frac{-\hh\Delta _\Pi}{4}}f$ (\cite{folland}
Proposition 2.97)\footnote{The variable of $f$ is not conjugated because we work in the
momentum representation.},
where $\Delta_\Pi$ is the Laplacian on functions,
\[
\Delta_\Pi =  -d^*\circ d, \quad d:C^\infty(\jacob)\to\Omega^1(\jacob).
\]
On a general Riemannian manifold this operator is given in
local coordinates by the formula
\begin{eqnarray*}
\Delta_{\Pi} f=\frac{1}{\sqrt{\det ({\bf g})}}\frac{\partial}{ \partial
    x^j}\left({\bf g}^{jk}
\sqrt{\det ({\bf g})}\frac{\partial f}{\partial x^k}\right),
\end{eqnarray*}
where ${\bf g}=({\bf g}_{jk})$ is the metric and ${\bf g}^{-1}=({\bf
  g}^{jk})$.
In the K\"ahler case, if the K\"ahler form is given in holomorphic coordinates by
\[
\omega = \frac{i}{2}\, \sum_{j,k} {\bf h}_{j k}\ dz_j\wedge d\bar{z}_k,
\]
then
\[
\Delta_\Pi = 4\sum_{j,k} {\bf h}^{j k}\ \frac{\partial ^2\ }{\partial z_j \partial \bar{z}_k},
\]
where $({\bf h}^{j k}) = ({\bf h}_{j k})^{-1}$.
In our situation, in the coordinates $z_j,\bar{z}_j$, $j=1,2,\ldots, g$, one
computes that $({\bf h}_{j k})^{-1}=Y^{-1}$
and therefore
$ ({\bf h}^{j k}) = Y$
(recall that $Y$ is the imaginary part of the matrix $\Pi$).
For Weyl quantization one introduces
a factor of $\frac{1}{2\pi}$ in front of the operator.  As such,
 the Laplace (or rather Laplace-Beltrami)  operator $\Delta_{\Pi}$ is equal to
\[
 \sum_{j,\,k=1}^g Y_{jk}
\left[(I_g+iY^{-1}X)\nabla_x-iY^{-1}\nabla_y\right]_j
\left[(I_g-iY^{-1}X)\nabla_x+iY^{-1}\nabla_y\right]_k.
\]
(A word about the notation being used:  $\nabla$ represents the usual (column) vector of partial
derivatives in the indicated variables, so that each object in the square brackets
is a column vector of partial derivatives.  The subindices $j, k$ are the corresponding
components of those vectors.)  A tedious calculation that we omit results in the
following formula for the Laplacian in the $(x,y)$ coordinates:
\begin{eqnarray*}
\Delta_\Pi = \sum (Y+ XY^{-1}X)_{jk}\,  &\frac{\partial^2\ }{\partial x_j \partial x_k}
-2 (XY^{-1})_{jk} \frac{\partial^2\ }{\partial x_j \partial y_k} + Y^{jk}
 \frac{\partial^2\ }{\partial y_j \partial y_k} .
\end{eqnarray*}
We will only need to apply $\Delta_\Pi$ explicitly to exponentials, as part of the
proof of the following basic proposition.  Note that the exponential function
\[
e^{2\pi i(p^Tx+q^Ty)}
\]
defines a function on the Jacobian provided $p, q\in {\mathbb Z}^g$.

\begin{proposition}\label{weylquantization}
The Weyl quantization of the exponentials is given by
\begin{eqnarray*}
\op\left(e^{2\pi i(p^Tx+q^Ty)}\right)\theta_\mu^\Pi(z)=
e^{-\frac{\pi i}{N}p^Tq-\frac{2\pi i}{N}\mu^Tq}\theta_{\mu+p}^\Pi(z).
\end{eqnarray*}
\end{proposition}

\begin{proof}
Let us introduce some useful notation local to the proof.  Note that $N$ and $\Pi$ are fixed throughout.
\begin{enumerate}
\item $e(t):= \exp(2\pi i Nt)$,
\item For $n\in{\mathbb Z}^g$ and $\mu\in\{0, 1, \ldots N-1\}^g$,
$n_\mu := n+\frac{\mu}{N}$.
\item $Q(n_\mu) := \frac{1}{2}(n_\mu^T\Pi n_\mu)$
\item $E_{p,q}(x,y) = e^{2\pi i(p^Tx+q^Ty)} = e(\frac{1}{N}(p^Tx+q^Ty))$.
\end{enumerate}
With these notations, in the $(x, y)$ coordinates
\[
\theta_\mu(x,y) = \sum_{n\in{\mathbb Z}^g} e(Q(n_\mu))\, e(n_\mu^T(x+\Pi y)).
\]
We first compute the matrix coefficients of the Toeplitz operator with symbol $E_{p,q}$, namely
$\langle E_{p,q}\theta_\mu\,,\,\theta_\nu\rangle$, which is
\begin{eqnarray*}
 (2N)^{g/2}\det(Y)^{1/2}\int_{[0,1]^{2g}} E_{p,q}(x,y)\theta_\mu(x,y)
\overline{\theta_\nu(x,y)}\, e^{-2\pi Ny^TYy}\, dxdy.
\end{eqnarray*}
Then a calculation shows that
\begin{eqnarray*}
&& E_{p,q}(x,y)\theta_\mu(x,y)
\overline{\theta_\nu(x,y)}\\
&& =\sum_{m,n\in {\mathbb Z}^g}e\Bigl[ Q(n_\mu) - \overline{Q(m_\nu)}+ (n_{\mu+p}-m_\nu)^Tx
+\bigl(\frac{q^T}{N}+n_\mu^T\Pi -m^T_\nu\overline{\Pi}\bigr)y\Bigr].
\end{eqnarray*}
The integral over $x\in [0,1]^g$ of the $(m,n)$ term will be non-zero iff
\[
N\Bigl(n_{\mu+p}-m_\nu\Bigr) = \mu+p-\nu + N(n-m)= 0,
\]
in which case the integral will be equal to one.  Therefore $\langle E_{p,q}\theta_\mu\,,\,\theta_\nu\rangle=0$
unless %
$[\nu] = [\mu +p]$,
where the brackets represent equivalence classes in ${\mathbb Z}^g_N$.  This shows that the
Toeplitz operator with multiplier $E_{p,q}$ maps $\theta_\mu$ to a scalar times $\theta_{\mu+p}$.  We now compute the scalar.

Taking $\mu$ in the fundamental domain $\{0,1,\cdots,N-1\}^g$ for ${\mathbb Z}^g_N$, there is a unique
representative, $\nu$, of $[\mu+p]$ in the same domain.  This $\nu$ is of the form
\[
\nu = \mu+p + N \kappa
\]
for a unique $\kappa\in {\mathbb Z}^g$.  With respect to the previous notation, $\kappa = n-m$.

It follows that
\begin{eqnarray*}
&&\langle E_{p,q}\theta_\mu\,,\,\theta_\nu\rangle
=(2N)^{g/2}\det(Y)^{1/2}\sum_{n\in{\mathbb Z}^g} \int_{[0,1]^{g}}
e\left[ Q(n_\mu) - \overline{Q(m_\nu)} \right.\\
&&\quad +\left.\bigl(\frac{q^T}{N}+n_\mu^T\Pi -m^T_\nu\overline{\Pi}\bigr)y+iy^TYy\right]\,dy,
\end{eqnarray*}
where $m  =n-\kappa$ in the $n$th term.
Using that $m_\nu = n_\mu + \frac{1}{N}p$, one obtains
\[
Q(n_\mu) - \overline{Q(m_\nu)} = i n_\mu^T Y n_\mu - \frac{1}{N}p^T \overline{\Pi}n_\mu
-\frac{1}{N^2}\overline{Q(p)}
\mbox{ and }
n_\mu^T\Pi -m^T_\nu\overline{\Pi} = 2in_\mu^TY -\frac{1}{N} p^T \overline{\Pi},
\]
and so we can write
\[
\begin{split}
\langle E_{p,q}\theta_\mu\,,\,\theta_\nu\rangle& =  (2N)^{g/2}\det(Y)^{1/2}\, e\Bigl[-\frac{1}{N^2}\overline{Q(p)}\Bigr]
\sum_{n\in{\mathbb Z}^g} \int_{[0,1]^{g}}\,dy \\ &
e\Bigl[ i n_\mu^T Y n_\mu - \frac{1}{N}p^T \overline{\Pi}n_\mu
 + \bigl(\frac{1}{N}
q^T+2in_\mu^TY -\frac{1}{N} p^T \overline{\Pi} \bigr) y + iy^TYy
 \Bigr].
\end{split}
\]
Making the change of variables $w:= y + n_\mu$ in the summand $n$, the argument of the function $e$  can be seen to be equal to
\[
iw^TYw + \frac{1}{N}\bigl(q^T-p^T\overline{\Pi}\bigr) w - \frac{1}{N}q^Tn_\mu.
\]
Since $q$ and $n$ are integer vectors,
\[
e\Bigl(\frac{1}{N} q^Tn_\mu\Bigr) = e^{-2\pi i q^T\mu/N}.
\]
The dependence on $n$ of the integrand is a common factor that comes
out of the summation sign. The series now is of integral over the translates
of $[0,1]^n$ that tile the whole space. Therefore $\langle E_{p,q}\theta_\mu\,,\,\theta_\nu\rangle$ is equal to
\[
(2N)^{g/2}\det(Y)^{1/2}\,e\Bigl[-\frac{1}{N^2}\overline{Q(p)}\Bigr]
e^{-2\pi i q^T\mu/N} \int_{{\mathbb R}^g} e^{-2\pi Nw^TYw + 2\pi i\bigl(q^T-p^T\overline{\Pi}\bigr) w}\, dw.
\]
A calculation of the integral\footnote{$\int_{{\mathbb R}^g} e^{-x^TAx + b^Tx}\,dx =
\Bigl(\frac{\pi^g}{\det A}\Bigr)^{1/2}\, e^{\frac{1}{4}b^TA^{-1}b}$}
yields that it is equal to
\[
\Bigl(\frac{1}{2N}\Bigr)^{g/2}\; \det (Y)^{-1/2}\;
e^{-\frac{\pi}{2 N}(q^T -p^T\overline{\Pi})Y^{-1}(q -\overline{\Pi}p)}.
\]
and so
\[
\langle E_{p,q}\theta_\mu\,,\,\theta_\nu\rangle =
  e^{-\frac{\pi i}{N} p^T\overline{\Pi}p}\;
e^{-2\pi i q^T\mu/N}\;e^{-\frac{\pi}{2N}(q^T -p^T\overline{\Pi})Y^{-1}(q -\overline{\Pi}p)}.
\]
The exponent on the right-hand side is $(-\pi/N)$ times
\begin{eqnarray*}
&&2iq^T\mu +ip^T(X-iY)p + \frac{1}{2}\Bigl([q^T - p^T(X-iY)]Y^{-1}[q-(X-iY)p]
\Bigr) \\
&&=2iq^T\mu +ip^T(X-iY)p + \frac{1}{2}\Bigl([ q^TY^{-1}-p^TXY^{-1}+ip^T][q-Xp+iYp]\Bigr)\\
&&=2iq^T\mu +ip^T(X-iY)p + \frac{1}{2}\Bigl(q^TY^{-1}q-2q^TY^{-1}Xp +2iq^Tp  \\
&& \quad + p^TXY^{-1}Xp - 2ip^TXp - p^TYp\Bigr) =2iq^T\mu + iq^Tp+ \frac{1}{2}{\mathcal R}
\end{eqnarray*}
where
\[
{\mathcal R} := q^TY^{-1}q-2q^TY^{-1}Xp+p^T( XY^{-1}X+Y )p.
\]
That is,
\begin{equation}\label{matcoeff}
\langle E_{p,q}\theta_\mu\,,\,\theta_\nu\rangle = e^{-\frac{2\pi i}{N}q^t\mu - \frac{\pi i}{N} q^Tp - \frac{\pi}{2N}{\mathcal R}}.
\end{equation}
\bigskip
On the other hand, it is easy to check that
$
\Delta_\Pi (E_{p,q}) = -(2\pi)^2 {\mathcal R} E_{p,q} ,
$
and therefore
\[
e^{-\frac{\Delta_\Pi}{4N}} (E_{p,q}) = e^{\frac{\pi}{2N}{\mathcal R}}\,E_{p,q},
\]
so that, by (\ref{matcoeff})
\[
\langle e^{-\frac{\Delta_\Pi}{4N}} (E_{p,q})\theta_\mu\,,\,\theta_\nu\rangle = e^{-\frac{2\pi i}{N}q^t\mu - \frac{\pi i}{N} q^Tp},
\]
as desired.
\end{proof}

\medskip
Let us focus on the group
of quantized exponentials. First note that the symplectic form $\omega$ induces
a nondegenerate bilinear form on ${\mathbb R}^{2g}$, which we denote also by
$\omega$, given by
\begin{eqnarray}\label{bilinform}
\omega((p,q),(p',q'))=\sum_{j=1}^g\left|
\begin{array}{cc}
p_j&q_j\\
p'_j&q'_j
\end{array}
\right|.
\end{eqnarray}

As a corollary of Proposition~\ref{weylquantization} we obtain the following
result.

\begin{proposition}\label{multofops}
Quantized exponentials satisfy the multiplication rule
\begin{eqnarray*}
&&\op\left(e^{2\pi i (p^Tx+q^Ty)}\right)\op \left(e^{2\pi i
    (p'^Tx+q'^Ty)}\right)\\
&& \quad =
e^{\frac{\pi i}{N}\omega((p,q),(p',q'))}\op(e^{2\pi i ((p+p')^Tx+(q+q')^Ty)}).
\end{eqnarray*}
\end{proposition}

This prompts us to define the Heisenberg group
\begin{eqnarray*}
\heis ({\mathbb Z}^g)=\{(p,q,k), \; p,q\in{\mathbb Z}^g, k\in {\mathbb Z}\}
\end{eqnarray*}
with multiplication
\begin{eqnarray*}
(p,q,k)(p',q',k')=(p+p',q+q',k+k'+\omega((p,q),(p',q'))).
\end{eqnarray*}
This group is a ${\mathbb Z}$-extension of $H_1(\Sigma_g,{\mathbb Z})$,
with the standard inclusion of $H_1(\Sigma_g,{\mathbb Z})$ into
it given by $$\sum p_ja_j+\sum q_kb_k\mapsto
(p_1,\ldots ,p_g,q_1,\ldots, q_g, 0).$$

The map
\begin{eqnarray*}
(p,q,k)\mapsto \op\left(e^{\frac{\pi i}{N}k}e^{ 2\pi i(p^Tx+q^Ty)}\right)
\end{eqnarray*}
defines a representation  of $\heis({\mathbb Z}^g)$ on theta functions.
To make this representation faithful, we factor it  by its kernel.

\begin{proposition}\label{kernel}
The set of elements in $\heis ({\mathbb Z}^g)$ that act on theta
functions as identity operators is the normal subgroup consisting of
the  $N$th powers of elements of the form $(p,q,k)$ with $k$ even.
The quotient group is isomorphic to a finite Heisenberg group.
\end{proposition}

Recall (cf. \cite{polishchuk}) that a finite Heisenberg group $H$ is a central extension
\begin{eqnarray*}
0\rightarrow {\mathbb Z}_{m}\rightarrow H\rightarrow K\rightarrow 0
\end{eqnarray*}
where $K$ is a finite abelian group such that the commutator pairing
$K\times K\rightarrow {\mathbb Z}_m$, $(k,k')\mapsto [\tilde{k},\tilde{k}']$ ($\tilde{k}$,
and $\tilde{k'}$ being arbitrary lifts of $k$ and $k'$ to $H$)
identifies $K$ with the group of homomorphisms from $K$ to ${\mathbb Z}_m$.

\begin{proof}
By Proposition~\ref{weylquantization},
\begin{eqnarray*}
(p,q,k)\theta_\mu^\Pi(z)=
e^{-\frac{\pi i}{N}p^Tq-\frac{2\pi i}{N}\mu^Tq+\frac{\pi}{N}k}\theta_{\mu+p}^\Pi(z).
\end{eqnarray*}
For $(p,q,k)$ to act as the identity operator, we should have
\begin{eqnarray*}
e^{-\frac{\pi i}{N}p^Tq-\frac{2\pi i}{N}\mu^Tq}\theta_{\mu+p}^\Pi(z)=\theta_{\mu}^\Pi(z)
\end{eqnarray*}
for all $\mu\in \{0,1,\ldots, N-1\}^g$. Consequently, $p$ should be
in $N{\mathbb Z}^g$. Then $p^Tq$ is a multiple of $N$, so the coefficient
  $e^{-\frac{\pi i}{N}p^Tq-\frac{2\pi i}{N}\mu^Tq+\frac{\pi i}{N}k}$
equals $\pm e^{-\frac{2\pi
      i}{N}\mu^Tq+\frac{\pi i}{N}k}$.
This coefficient should be equal to $1$. For $\mu=(0,0,\ldots, 0)$
this implies that $-p^Tq+k$ should be an even multiple of $N$.
But then by varying $\mu$
we conclude that $q$ is a multiple of $N$. Because $N$ is even, it
follows that $p^Tq$ is an even multiple of $N$, and consequently $k$ is
an even multiple of $N$. Thus any element in the kernel of the representation
must belong to $N{\mathbb Z}^{2g}\times (2N){\mathbb Z}$. It is easy to see
that any element of this form is in the kernel. These are precisely
the elements of the form $(p,q,k)^N$ with $k$ even.

The quotient of $\heis ({\mathbb Z}^g)$ by the kernel of the representation
is a ${\mathbb Z}_{2N}$-extension of the finite abelian group
${\mathbb Z}_N^{2g}$, thus is a finite Heisenberg group. This group
is isomorphic to
\begin{eqnarray*}
\{(p,q,k)\,|\, p,q\in {\mathbb Z}_N^g, k\in {\mathbb Z}_{2N}\}
\end{eqnarray*}
with the multiplication rule
\begin{eqnarray*}
(p,q,k)(p',q',k')=(p+p',q+q',k+k'+2pq').
\end{eqnarray*}
The isomorphism is induced by the map $F:\heis({\mathbb Z}^g)\rightarrow
{\mathbb Z}_N^{2g}\times {\mathbb Z}_{2N}$,
\begin{eqnarray*}
F(p,q,k)=(p\mbox{ mod }N, q\mbox{ mod }N, k+pq \mbox{ mod }2N).\qedhere
\end{eqnarray*}
\end{proof}

We denote by $\heis({\mathbb Z}_N^g)$ this finite Heisenberg group and
by $\exp(p^TP+q^TQ+kE)$  the image
of 
$(p,q,k)$ in it. The representation of $\heis({\mathbb Z}_N^g)$
on the space of theta functions is called the {\em Schr\"{o}dinger
representation}. It is an analogue, for the case of the $2g$-dimensional torus,
of the standard Schr\"{o}dinger representation of the Heisenberg
group with real entries on $L^2({\mathbb R})$.
In particular we have
\begin{eqnarray}\label{schroedinger}
\left.
\begin{array}{l}
\exp(p^TP)\theta_\mu^\Pi(z)=\theta_{\mu+p}^\Pi(z)\\
\exp(q^TQ)\theta_\mu^\Pi(z)=e^{-\frac{2\pi i}{N}q^T\mu}\theta_{\mu}^\Pi(z)\\
\exp(kE)\theta_\mu^\Pi(z)=e^{\frac{\pi i}{N}k}\theta_{\mu}^\Pi(z).
\end{array}
\right.
\end{eqnarray}

\begin{theorem}\label{svn}
(Stone-von Neumann)  The Schr\"{o}dinger representation
of $\heis ({\mathbb Z}_N^g)$ is the {\em unique} (up to an isomorphism) irreducible
unitary representation of this group with the property that
$\exp(kE)$ acts as $e^{\frac{\pi i}{N}k}Id$  for all
$k\in {\mathbb Z}$.
\end{theorem}

\begin{proof}
Let $X_j=\exp(P_j)$, $Y_j=\exp(Q_j)$,  $j=1,2,\ldots, g$, $Z=\exp(E)$.
Then $X_jY_j=Z^2Y_jX_j$, $X_jY_k=Y_kX_j$ if $j\neq k$,  $X_jX_k=X_kX_j$,
$Y_jY_k=Y_kY_j$, $ZX_j=X_jZ$, $ZY_j=Y_jZ$,  for
all $i,j$, and $X_j^N=Y_j^N=Z^{2N}=Id$ for all $j$. Because $Y_1,Y_2,\ldots, Y_g$
commute pairwise, they have a common eigenvector $v$. And because
$Y_j^N=Id$ for all $j$, the eigenvalues $\lambda_1,\lambda_2,\ldots,
\lambda_g$ of $v$ with respect to the $Y_1,Y_2,\ldots, Y_g$
are roots of unity. The  equalities
\begin{eqnarray*}
&& Y_jX_jv=e^{-\frac{2\pi i}{N}}X_jY_j=e^{-\frac{2\pi i}{N}}\lambda_jX_jv,\\
&& Y_jX_kv=X_kY_jv=\lambda_jX_kv, \quad \mbox{if }j\neq k
\end{eqnarray*}
show that by applying $X_j$'s repeatedly we can produce an eigenvector
$v_0$ of the commuting system $Y_1,Y_2,\ldots, Y_g$
 whose eigenvalues are all equal to $1$. The irreducible representation
is spanned by the vectors
$X_1^{n_1}X_2^{n_2}\cdots X_g^{n_g}v_0$, $n_i\in \{0,1,\ldots,
N-1\}$.
Any such vector is an eigenvector of the system $Y_1,Y_2,\ldots, Y_g$, with
eigenvalues respectively $e^{\frac{2\pi i}{N}n_1},$ $ e^{\frac{2\pi
    i}{N}n_2},\ldots,$ $e^{\frac{2\pi i}{N}n_g}$. So these vectors are linearly
 independent and form a basis of the irreducible representation. It is
not hard to see that the action of $\heis ({\mathbb Z}_N^g)$ on the vector
space spanned by these vectors is
the Schr\"{o}dinger representation.
\end{proof}

\begin{proposition}\label{allspace}
The operators $\op\left(e^{ 2\pi i(p^Tx+q^Ty)}\right)$, $p,q\in \{0,1,\ldots,
N-1\}^g$  form a basis of the space of linear operators on
$\spacetheta(\Sigma_g)$.
\end{proposition}

\begin{proof}
For simplicity, we show that the operators
\begin{eqnarray*}
e^{\frac{\pi i}{N}p^Tq}\op\left(e^{ 2\pi i(p^Tx+q^Ty)}\right), \quad  p,q\in \{0,1,\ldots,
N-1\}^g,
\end{eqnarray*}
 form a basis. Denote by $M_{p,q}$ the respective matrices of
these operators in the basis $(\theta_\mu^\Pi)_\mu$. For a fixed
$p$, the nonzero entries of the matrices $M_{p,q}$, $q\in \{0,1,\ldots, N-1\}^g$
are precisely those in the slots $(m,m+p)$, with $m\in \{0,1,\ldots, N-1\}^g$
(here $m+p$ is taken modulo $N$). If we vary $m$ and $q$ and arrange
these nonzero entries in a matrix, we obtain the $g$th power of a Vandermonde
matrix, which is nonsingular. We conclude that for fixed $p$, the matrices
$M_{p,q}$, $q\in \{0,1,\ldots, N-1\}^g$ form a basis for the vector
space of matrices with nonzero entries in the slots of the form $(m,m+p)$.
Varying $p$, we obtain the desired conclusion.
\end{proof}

\begin{cor}\label{groupalgheis}
The algebra $L(\spacetheta(\Sigma_g))$ of linear operators on
the space of theta functions is isomorphic to the algebra obtained
by factoring ${\mathbb C}[\heis({\mathbb Z}_N^g)]$ by the relation
$(0,0,1)=e^{\frac{i\pi}{N}}$.
\end{cor}

Let us now recall the action of the {\em modular group} on theta functions.
The modular group, known also as the {\em mapping class group},
of a simple closed surface  $\Sigma_g$ is the quotient of the group
of homemorphisms of $\Sigma_g$ by the subgroup of homeomorphisms that
are isotopic to the identity map.
 It is at this point where it is essential that
$N$ is  {\em even}.

The mapping class group acts on the Jacobian in the following way.
An element $h$ of this  group
induces a linear automorphism $h_*$ of
$H_1(\Sigma_g, {\mathbb R})$. The matrix of
$h_*$ has integer entries,  determinant $1$, and satisfies $h_*J_0{h_*}^T=J_0$,
where $J_0=\left(\begin{array}{cc}0&I_g\\I_g&0\end{array}\right)$ is
the intersection form in $H_1(\Sigma_g,{\mathbb R})$. As such, $h_*$
is a {\em symplectic linear automorphism} of $H_1(\Sigma_g,{\mathbb R})$, where
the symplectic form is the intersection form.
Identifying $\jacob$ with $H_1(\Sigma_g,{\mathbb R})/
H_1(\Sigma_g,{\mathbb  Z})$, we see that $h_*$ induces a symplectomorphism
$\tilde{h}$ of $\jacob$. The map $h\rightarrow \tilde{h}$ induces
an action of the mapping class group of $\Sigma_g$ on the Jacobian variety.
This action can be described explicitly as follows.
Decompose $h_*$ into $g\times g$
blocks as \begin{eqnarray*}
h_*=\left(\begin{array}{cc}A&B\\C&D\end{array}\right).
\end{eqnarray*}
Then  $\tilde{h}$ maps the complex torus defined by the lattice
$(I_g,\Pi)$ and complex variable $z$
to the complex torus defined by the lattice $(I_g,\Pi')$ and complex variable
$z'$,
where $\Pi'=(\Pi C+D)^{-1}(\Pi A+B)$ and $z'=(\Pi C+D)^{-1}z$.

This action of the mapping class group of the surface on the Jacobian
induces an action of the mapping class group on the finite Heisenberg
group by
\begin{eqnarray*}
h\cdot \exp(p^TP+q^TQ+kE)=\exp[(Ap+Bq)^TP+(Cp+Dq)^TQ+kE].
\end{eqnarray*}
The nature of this action is as follows: Since $h$ induces a diffeomorphism
on the Jacobian, we can compose $h$ with an exponential and then quantize;
the resulting operator is as above.
We point out that if $N$ were not even, this action would
be defined only for $h_*$  in the subgroup $Sp_{\theta}(2n,{\mathbb Z})$
of the  symplectic group (this is because only for $N$ even is the
kernel of the map $F$ defined in Proposition~\ref{kernel}
 preserved under the action of $h_*$).

As a corollary of Theorem~\ref{svn}, the representation
of the finite Heisenberg group on theta functions
given by $u\cdot \theta_\mu^\Pi=(h\cdot u)\theta_\mu^\Pi$ is
equivalent to the Schr\"{o}dinger representation, hence there
is an automorphism $\rho(h)$ of $\spacetheta(\Sigma_g)$
that satisfies  the {\em exact Egorov identity}:
\begin{eqnarray}\label{egorov}
h\cdot \exp(p^TP+q^TQ+kE)=\rho(h)\exp(p^TP+q^TQ+kE)\rho(h)^{-1}.
\end{eqnarray}
(Compare with \cite{folland}, Theorem 2.15, which is the analogous statement in quantum mechanics in Euclidean space.)
Moreover, by Schur's lemma,  $\rho(h)$ is unique up to multiplication
by a constant. We thus have a projective representation
of the mapping class group of the surface on the space of classical
 theta functions
 that statisfies with the action of the finite
Heisenberg group the exact Egorov identity from (\ref{egorov}). This
is the finite dimensional counterpart of the metaplectic representation,
called the {\em Hermite-Jacobi action}.

\begin{remark}
We emphasize  that the action of the mapping class group
of $\Sigma_g$ on theta functions factors through an action of the
symplectic group $Sp(2n, {\mathbb Z})$.
\end{remark}

Up to multiplication by a constant,
\begin{eqnarray}\label{hermitejacobi}
\rho(h)\theta_\mu^{\Pi}(z)=
\exp[-\pi i
z^TC(\Pi C+D)^{-1}z]\theta_\mu^{\Pi'}(z')
\end{eqnarray}
(cf. (5.6.3) in \cite{polishchuk}).  When the Riemann surface
is the complex torus obtained as
the quotient of the complex plane by the integer lattice,
and $h=S$ is  the map induced by a $90^\circ$ rotation around the origin,
then $\rho(S)$ is the discrete Fourier transform. In general,
like for the metaplectic representation
(see \cite{lionvergne}),  $\rho(h)$ can be written as a composition of partial
discrete Fourier transforms. For this reason, we will  refer
to $\rho(h)$ as a {\em discrete Fourier transform}.

\section{Theta functions in the abstract setting}\label{sec:3}

In this section we apply to the finite Heisenberg group the standard
construction which identifies the Schr\"{o}dinger representation as
a representation induced by an irreducible representation (i.e. character)
of a maximal abelian subgroup (see for example \cite{lionvergne}).

 Start with a Lagrangian subspace
 of $H_1(\Sigma_g,{\mathbb R})$ with respect to the intersection form,
which for our purpose is  spanned by the elements
$b_1,b_2,\ldots, b_g$ of the canonical basis.
Let ${\bf L}$ be the intersection of this space with
$H_1(\Sigma_g,{\mathbb Z})$. Under the standard inclusion
$H_1(\Sigma_g,{\mathbb Z})\subset \heis({\mathbb Z}^g)$,  ${\bf L}$
becomes an abelian subgroup of the Heisenberg group with integer
entries. This factors to an abelian subgroup $\exp({\bf L})$
of $\heis({\mathbb Z}_N^g)$.
Let $\exp({\bf L}+{\mathbb Z}E)$ be the
subgroup of $\heis({\mathbb Z}_N^g)$
containing both $\exp({\bf L})$ and the scalars $\exp({\mathbb Z}E)$.
Then $\exp({\bf L}+{\mathbb Z}E)$ is a maximal abelian subgroup.
Being abelian, it has only 1-dimensional irreducible representations,
which are its  characters.

In view of the Stone-von Neumann Theorem, we consider the
induced representation defined by the  character $\chi_{\bf L}:\exp({\bf L}+{\mathbb Z}E)\rightarrow
{\mathbb C}$,
$\chi_{\bf L}(l+kE)=e^{\frac{\pi i}{N}k}$.
This representation is
\begin{eqnarray*}
\mbox{Ind}_{\exp({\bf L}+{\mathbb Z}E)}^{\heis({\mathbb Z}_N^g)}=C[\heis(
{\mathbb Z}_N^g)]\bigotimes_{{\mathbb C}[\exp({\bf L}+{\mathbb Z}E)]}{\mathbb C}
\end{eqnarray*}
with $\heis({\mathbb Z}_N^g)$ acting on the left in the first factor of
the tensor product. Explicitly, the vector space of the representation is
the quotient of the group algebra ${\mathbb C}[\heis({\mathbb Z}_N^g)]$ by
the vector subspace spanned by all elements of the form
\begin{eqnarray*}
u-\chi_{\bf L}(u')^{-1}uu'
\end{eqnarray*} with $u\in \heis ({\mathbb Z}_N^g)$
and $u'\in\exp({\bf L}+{\mathbb Z}E)$. We denote this quotient by
${\mathcal H}_{N,g}({\bf L})$, and let $\pi_{\bf L}:{\mathbb C}[\heis({\mathbb Z}_N^g)]\rightarrow {\mathcal H}_{N,g}({\bf L})$ be the quotient map.
Let also the inner product be defined such that $\pi_{\bf L}(u)$ has norm
$1$, where $u$ is an element of the finite Heisenberg group seen as an element
of its group algebra.

The left regular action of the  Heisenberg group $\heis({\mathbb Z}_N^g)$
on its group algebra descends to an action on ${\mathcal H}_{N,g}({\bf L})$.

\begin{proposition}\label{abstracttheta}
The map $\theta_\mu^{\Pi}(z)\mapsto \pi_{\bf L}(\exp(\mu^TP))$,
 $\mu\in {\mathbb Z}_N^g$
defines a unitary map  between the space of theta functions
$\spacetheta(\Sigma_g)$
and  ${\mathcal H}_{N,g}({\bf L})$, which intertwines the Schr\"{o}dinger
representation and the left action of the finite Heisenberg group.
\end{proposition}

\begin{proof}
It is not hard to see that $\spacetheta(\Sigma_g)$ and
${\mathcal H}_{N,g}({\bf L})$
have the same dimension. Also, for $\mu\neq \mu'\in {\mathbb Z}_N^g$,
$\pi_{\bf L}(\exp(\mu^TP))\neq \pi_{\bf L}(\exp(\mu'^TP))$,
hence the map from the statement
is an isomorphism of finite dimensional spaces. The
norm of $\pi_{\bf L}(\exp(\mu^TP))$ is one, hence this map is unitary.
We have
\begin{eqnarray*}
\exp(p^TP)\exp(\mu^TP)=\exp((p+\mu)^T)P)
\end{eqnarray*}
and
\begin{eqnarray*}
\exp(q^TQ)\exp(\mu^TP)=e^{-\frac{\pi i }{N}q^T\mu}\exp(\mu^TP)\exp(q^TQ).
\end{eqnarray*}
It follows that
\begin{eqnarray*}
&&\exp(p^TP)\pi_{\bf L}(\exp(\mu^TP))=\pi_{\bf L} ((p+\mu)^TP)\\
&&\exp(q^TQ)\pi_{\bf L}(\exp(\mu^TP))=e^{-\frac{\pi i}{N}q^T\mu} \pi_{\bf L}(\exp(\mu^TP))
\end{eqnarray*}
in agreement with the Schr\"{o}dinger representation (\ref{schroedinger}).
\end{proof}

We rephrase the Hermite-Jacobi action
in this setting. To this end, fix an element $h$ of the mapping class group of
the Riemann surface $\Sigma_g$. Let ${\bf L}$  be
the  subgroup of $H_1(\Sigma_g,{\mathbb Z})$  associated to a canonical basis as
explained in the beginning of this section, which determines the
maximal abelian subgroup $\exp({\bf L}+{\mathbb Z}E)$.

The automorphism of $H_1(\Sigma_g,{\mathbb Z})$ defined by $a_j\mapsto h_*(a_j)$, $b_j\mapsto h_*(b_j)$, $j=1,2,\ldots, g$, maps isomorphically
${\bf L}$ to
$h_*({\bf L})$, and thus allows us to identify in a canonical
fashion ${\mathcal H}_{N,g}({\bf L})$ and ${\mathcal H}_{N,g}(h_*({\bf L}))$.
Given this identification, we can view the discrete Fourier tranform
as a map
$\rho(h):{\mathcal H}_{N,g}({\bf L})\rightarrow {\mathcal H}_{N,g}(h_*({\bf L}))$.

The discrete Fourier transform should map an
element ${\bf u}\mbox{ mod }\mbox{ker}(\pi_{\bf L})$ in  the space
$ {\mathbb C}[\heis({\mathbb Z}_N^g)]/\mbox{ker}(\pi_{\bf L})$ to
${\bf u}\mbox{ mod }\mbox{ker}(\pi_{h_*({\bf L})})$ in
${\mathbb C}[\heis ({\mathbb Z}_N^g)] /\mbox{ker}(\pi_{h_*({\bf L})}).$
In this form the map is not well defined, since different representatives
for the class of ${\bf u}$ might yield different images. The idea
is to consider all possible liftings of ${\bf u}$ and average them.
For lifting the element ${\bf u} \mbox{ mod }\mbox{ker}(\pi_{\bf L})$ we
use the section of $\pi_{\bf L}$ defined as
\begin{eqnarray}\label{lift}
s_{\bf L}({\bf u} \mbox{ mod }\mbox{ker}(\pi_{\bf L}))=
\frac{1}{2N^{g+1}}
\sum_{u_1\in\exp ({\bf L}+{\mathbb Z}E)}
\chi_{\bf L}(u_1)^{-1}{\bf u}u_1.
\end{eqnarray}
Then, up to multiplication by a constant
\begin{eqnarray}\label{fouriertransform}
\quad\quad \rho(h)({\bf u} \,\mbox{mod}\,\mbox{ker}(\pi_{\bf L}))
=\frac{1}{2N^{g+1}}\!\!
\sum_{u_1\in\exp ({\bf L}+{\mathbb Z}E)}\!\!
\chi_{\bf L}(u_1)^{-1}{\bf u}u_1\, \mbox{mod}\,\mbox{ker}(\pi_{h_*(\bf L)}).
\end{eqnarray}
This formula identifies $\rho(h)$ as a  Fourier transform. 
 That this map agrees with the one defined
by (\ref{hermitejacobi}) up to multiplication by a constant follows
from Schur's lemma, since both maps satisfy the exact Egorov identity
 (\ref{egorov}).



\section{A topological  model for theta functions}\label{sec:4}

The  finite Heisenberg
group, the equivalence relation defined by the kernel of $\pi_{\bf L}$,
and the Schr\"{o}dinger representation can be given topological
interpretations, which we  explicate below. First, a heuristical discussion.

\medskip

{\em The Heisenberg group}. The group $\heis({\mathbb Z}^{g})$
is a ${\mathbb Z}$-extension of the abelian group
$H_1(\Sigma_g,{\mathbb Z})$. The bilinear form $\omega$ from (\ref{bilinform}),
which defines the cocycle of this extension, is the intersection form
in $H_1(\Sigma_g,{\mathbb Z})$. Cycles in $H_1(\Sigma_g,{\mathbb Z})$
can be represented by families of non-intersecting
 simple closed curves on the surface.
As vector spaces, we can identify
${\mathbb C}[\heis ({\mathbb Z}^g)]$ with ${\mathbb C}[t,t^{-1}]H_1(\Sigma_g,
{\mathbb Z})$, where $t$ is an abstract variable whose exponent equals the
last coordinate in the Heisenberg group.

We start with
an  example on the torus. Here and throughout the paper
we agree  that $(p,q)$ denotes the curve of slope $q/p$ on the torus,
oriented from the origin to the point $(p,q)$ when viewing the torus
as a quotient of the plane by integer translations.
Consider  the multiplication
\begin{eqnarray*}
(1,0)(0,1)=t(1,1),
\end{eqnarray*}
shown  in  Figure~\ref{torusmult}.
The product curve $(1,1)$
can be obtained by cutting open the  curves $(1,0)$ and $(0,1)$
 at the crossing
and joining the ends so that the orientations agree. This
 operation  is called {\em smoothing of the crossing}.
It is easy
to check that this works   for arbitrary surfaces:
whenever  multiplying two
families of curves introduce a coefficient of $t$ raised to the
algebraic intersection number of the two families then smoothen all crossings.
Such algebras of curves, with multiplication related to
polynomial invariants of knots, were first considered in \cite{turaev2}.

\begin{figure}[ht]
\centering
\resizebox{.40\textwidth}{!}{\includegraphics{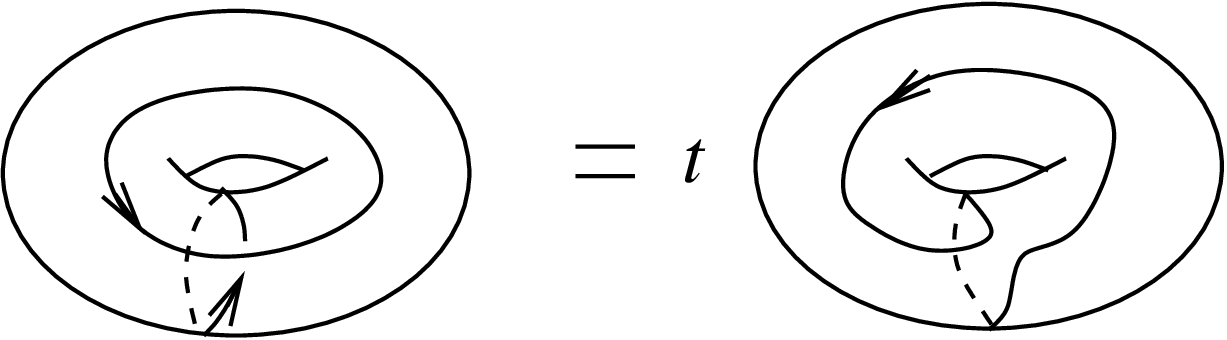}}
\caption{}
\label{torusmult}
\end{figure}

The group $\heis ({\mathbb Z}_N^g)$ is a quotient
of $ \heis ({\mathbb Z}^g)$, but can also be viewed as an extension
of $H_1(\Sigma_g,{\mathbb Z}_N)$. As such, the elements of
${\mathbb C}[\heis({\mathbb Z}_N^g)]$ can be represented by
families of non-intersecting simple closed curves on the surface with
the convention that any $N$ parallel curves can be deleted.
The above observation applies to this case as well, provided that
we set $t=e^{\frac{i\pi}{N}}$.

It follows that the space of linear operators $L(\spacetheta(\Sigma_g))$ can
be represented as an algebra of simple closed curves on the surface
with the convention that any $N$ parallel curves can be deleted. The
multiplication of two families of simple closed curves
is defined by introducing a coefficient of $e^{\frac{i\pi}{N}}$ raised
to the algebraic intersection number of the two families and smoothing
the crossings.

\medskip

{\em Theta functions.}
Next, we examine the space of theta functions, in its abstract
framework from Section~\ref{sec:3}. To better understand the factorization
modulo the kernel of $\pi_{\bf L}$, we look again at
the torus. If the canonical basis is $(1,0)$ and $(0,1)$ with ${\bf L}=
{\mathbb Z}(0,1)$ , then an
equivalence modulo $\mbox{ker}(\pi_{\bf L})$ is shown in
Figure~\ref{equivmodL}. If we map the torus to the boundary of a solid
torus in such a way that ${\bf L}$ becomes null-homologous, then
the first and  last curves from Figure~\ref{equivmodL} are homologous
in the solid torus. To keep track of $t$ we apply a standard method in
topology which consists of framing the curves. A framed curve in a manifold
is an embedding of an annulus. One can think of the curve as being one
of the boundary components of the annulus, and then the annulus itself
keeps track of the number of ways that the curve twists around itself.
Changing the framing by a full twist amounts to multiplying by
$t$ or $t^{-1}$ depending whether the twist is  positive or negative.
Then the equality from Figure~\ref{equivmodL} holds in the solid torus.
It is not hard to check for a general surface $\Sigma_g$ the equivalence
relation modulo $\mbox{ker}(\pi_{\bf L})$ is of this form in the
handlebody bounded by $\Sigma_g$  in such a way that ${\bf L}$ is
null-homologous.

\begin{figure}[ht]
\centering
\resizebox{.55\textwidth}{!}{\includegraphics{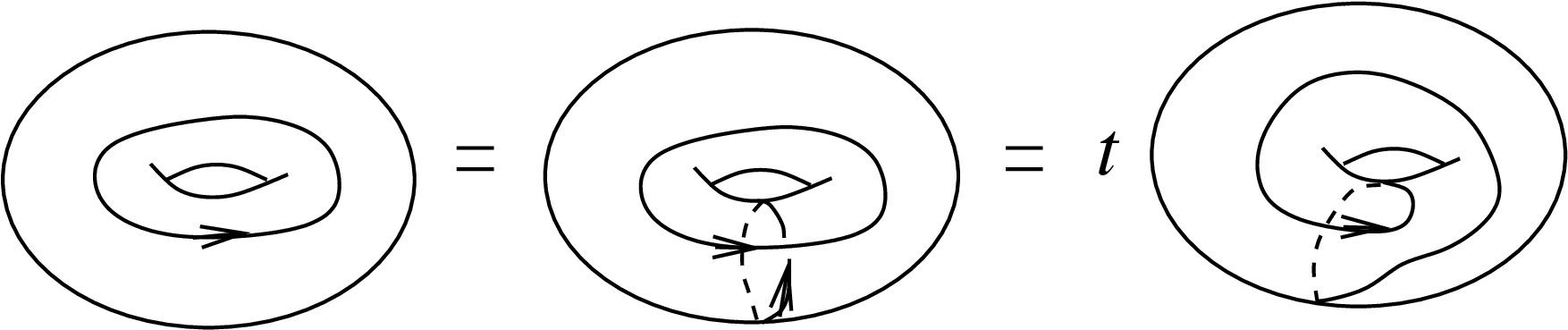}}
\caption{}
\label{equivmodL}
\end{figure}

\medskip

{\em The Schr\"{o}dinger representation.}
One can frame the curves on $\Sigma_g$ by using the blackboard
framing, namely by embedding the annulus in the surface.
As such, the Schr\"{o}dinger representation is the left action of
an algebra of framed curves on a surface on the vector space of
framed curves in the handlebody induced by the inclusion of the surface
in  the handlebody. We will make this precise using
the language of skein modules \cite{przytycki}.

\medskip

Let $M$ be a compact oriented $3$-dimensional manifold. A framed
link in $M$ is a smooth embedding of a disjoint union of finitely many
annuli. The annuli are called link components. We consider
oriented framed links. The orientation of a  link component is an
orientation of one of the
circles that bound the annulus.  When $M$ is the cylinder
over a surface, we  represent framed links as oriented curves
with the blackboard framing, meaning that
the annulus giving the framing is always parallel to the surface.

Let $t$ be a free variable. Consider the free
${\mathbb C}[t,t^{-1}]$-module with basis
the isotopy classes of framed oriented links in $M$ including
the empty link $\emptyset$. Let ${\mathcal S}$ be the
 the submodule spanned by all elements of the form
depicted in  Figure~\ref{skeinrelations}, where the two terms in
each skein relation depict framed links that are identical except in
an embedded ball, in which they look as shown.
The ball containing the crossing can be embedded in any
possible way.
To normalize, we add to ${\mathcal S}$ the element consisting
of the difference between the unknot in $M$ and the empty link
$\emptyset$. Recall that the  unknot is an embedded circle
that bounds an embedded disk in $M$ and whose framing annulus lies inside the
disk.

\begin{definition}
The result of the factorization of the free  ${\mathbb C}[t,t^{-1}]$-module
with basis the isotopy classes of framed oriented links by the
submodule ${\mathcal S}$
is called the  {\em linking number skein module} of $M$, and is
 denoted by ${\mathcal L}(M)$. The elements of ${\mathcal L}(M)$
are called {\em skeins}.
\end{definition}

 In other words, we
are allowed to smoothen each crossing, to change the framing
provided that we multiply by the appropriate power of $t$, and
to identify  the unknot with the empty link.

\begin{figure}[ht]
\centering
\resizebox{.40\textwidth}{!}{\includegraphics{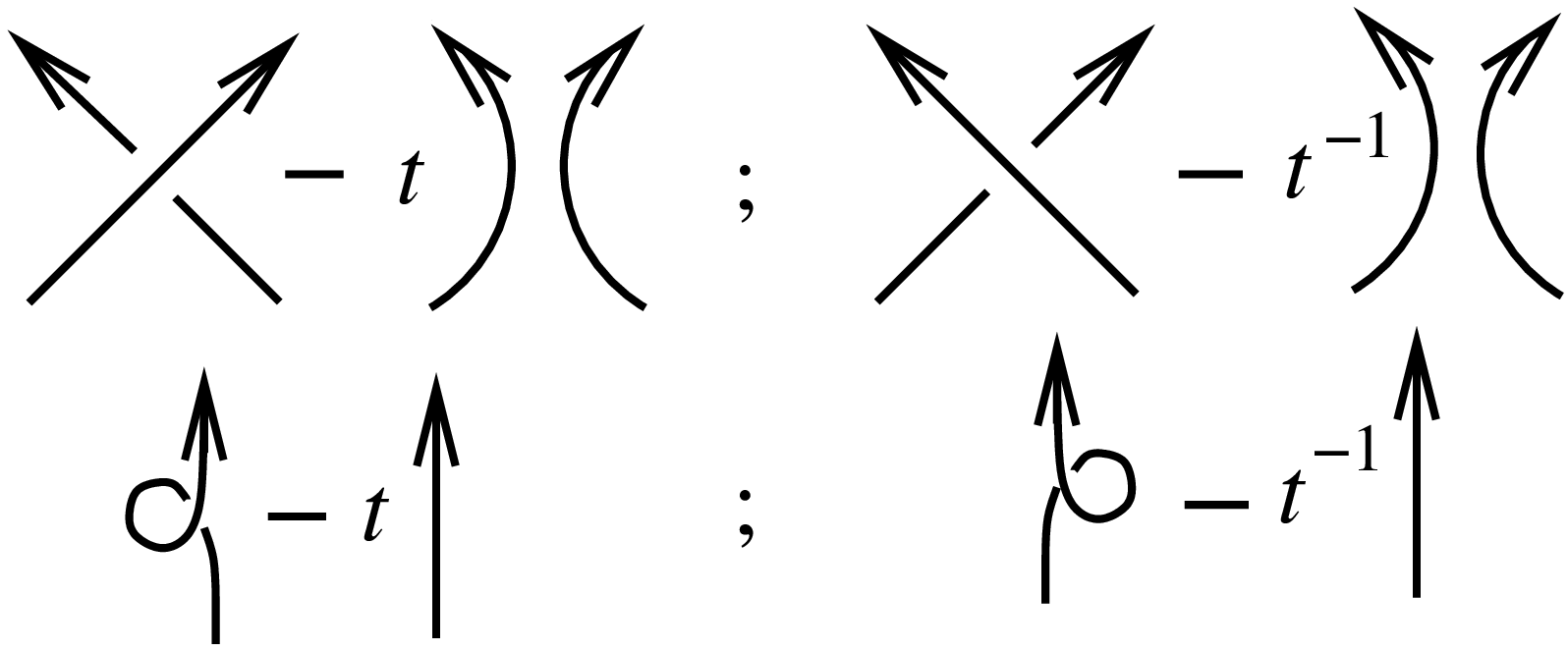}}
\caption{}
\label{skeinrelations}
\end{figure}

The ``linking number'' in the  name is motivated by the fact that
 the skein relations from
Figure~\ref{skeinrelations} are used for computing the linking number.
These skein modules were first introduced by Przytycki in \cite{przytycki2}
as one-parameter deformations of the group algebra of $H_1(M,{\mathbb Z})$.
Przytycki computed them for all $3$-dimensional manifolds.

\begin{lemma}
Any trivial link component, namely
any link component  that bounds a disk disjoint from the rest of the link in
such a way that the framing is an annulus inside the disk, can be deleted.
\end{lemma}

\begin{proof}
The proof of the lemma is given in Figure~\ref{triviallink}.
\end{proof}

\begin{figure}[h]
\centering
\resizebox{.25\textwidth}{!}{\includegraphics{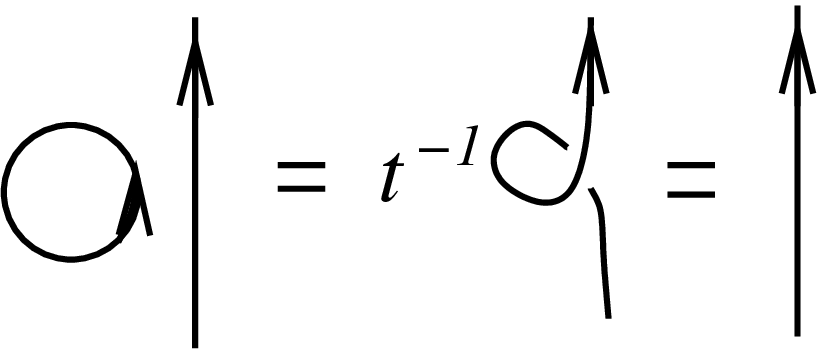}}
\caption{}
\label{triviallink}
\end{figure}

If $M=\Sigma_g\times[0,1]$, the cylinder over a surface, then the
identification
\begin{eqnarray*}
\Sigma_g\times [0,1]\cup\Sigma_g\times [0,1]\approx \Sigma\times [0,1]
\end{eqnarray*}
obtained by gluing the boundary component
$\Sigma_g\times \{0\}$ in the first cylinder to the boundary
component $\Sigma_g\times \{1\}$ in the second cylinder by the
identity map
induces a multiplication on ${\mathcal L}(\Sigma_g\times [0,1])$.
This turns ${\mathcal L}(\Sigma_g\times [0,1])$  into an
algebra, called the {\em linking number skein algebra}.
As such, the product
of two skeins is obtained by placing the first  skein on top of the second.
The $n$th power of an oriented, framed, simple closed curve consists
then of $n$ parallel copies of that curve. We adopt the same terminology
even if the manifold is not a cylinder, so $\gamma^n$ stands
for $n$ parallel copies of $\gamma$. Additionally, $\gamma^{-1}$
is obtained from  $\gamma $ by  reversing orientation, and
$\gamma^{-n}=(\gamma^{-1})^n$.

\begin{definition}
For a fixed positive integer $N$, we define the {\em reduced linking
number skein module}  of the manifold $M$,
denoted by ${\mathcal L}_N(M)$,
to be the quotient of ${\mathcal L}(M)$ obtained by imposing
that $\gamma^N=\emptyset$ for every oriented, framed, simple closed
curve $\gamma$, and by setting $t=e^{\frac{\pi i}{N}}$. As such, $L=L'$ whenever
$L'$ is obtained from $L$ by removing $N$ parallel link components.
\end{definition}

\begin{remark}
As a rule followed
throughout the paper, whenever we talk about skein modules, $t$ is
a free variable, and when we talk about reduced skein modules,
$t$ is a root of unity.
The isomorphisms ${\mathcal L}(S^3)\cong {\mathbb C}[t,t^{-1}]$
and ${\mathcal L}_N(S^3)\cong {\mathbb C}$ allow us to
identify the linking number skein module of $S^3$ with the set of Laurent
polynomials in $t$ and the reduced skein module with ${\mathbb C}$.
\end{remark}

For a closed, oriented,
genus $g$ surface  $\Sigma_g$, consider a canonical basis
of its first homology $a_1,a_2,\ldots, a_g,b_1,b_2,\ldots, b_g$ (see
Section~\ref{sec:1}). The basis elements are oriented simple closed
curves on the surface, which we endow with the blackboard framing.
Let $H_g$ be a genus $g$ handlebody and $h_0:\Sigma_g\rightarrow \partial H_g$
be a homeomorphism that maps $b_1,b_2,\ldots ,b_g$ to null homologous curves.
Then $a_1,a_2,\ldots a_g$ is a basis of the first homology of the
handlebody. Endow these curves in the handlebody with the framing they had on
the surface.

 The linking number skein module of a $3$-manifold $M$ with boundary
is a module over the skein algebra of a boundary component $\Sigma_g$.
The module structure is induced by the identification
\begin{eqnarray*}
\Sigma_g\times [0,1]\cup M\approx M
\end{eqnarray*}
where $\Sigma_g\times [0,1]$ is glued to $M$ along $\Sigma_g\times \{0\}$
by the identity map. This means that the module structure is
induced by identifying $\Sigma_g\times [0,1]$ with a regular
neighborhood of the boundary of $M$.
The product of a skein in a regular neighborhood
of the boundary and a skein in the interior is the union of
the two skeins. This module structure descends to relative skein modules.

In particular ${\mathcal L}(\Sigma_g\times [0,1])$ acts on the left on
${\mathcal L}(H_g)$ with action induced by the homeomorpism
$h_0:\Sigma_g\rightarrow \partial H_g$, and the action descends to relative
skein modules.

\begin{theorem}\label{heisenbergskeins1}
(a) The linking number skein module ${\mathcal L}(\Sigma_g\times [0,1])$
is a free ${\mathbb C}[t,t^{-1}]$-module with basis
\begin{eqnarray*}
a_1^{m_1}a_2^{m_2}\cdots a_g^{m_g}b_1^{n_1}b_2^{n_2}\cdots b_g^{n_g}, \quad
m_1,m_2,\ldots, m_g,n_1,n_2,\ldots, n_g\in {\mathbb Z}.
\end{eqnarray*}
(b) The linking number skein module ${\mathcal L}(H_g)$ is a free
${\mathbb C}[t,t^{-1}]$-module with basis
\begin{eqnarray*}
a_1^{m_1}a_2^{m_2}\cdots a_g^{m_g},\quad m_1,m_2,\ldots, m_g\in {\mathbb Z}.
\end{eqnarray*}
(c) The algebras
${\mathcal L}(\Sigma_g\times [0,1])$ and
${\mathbb C}[\heis({\mathbb Z}^g)]$ are isomorphic, with the
isomorphism defined by the map
\begin{eqnarray*}
t^k\gamma\mapsto ([\gamma],k).
\end{eqnarray*}
where $\gamma$ ranges over all skeins represented by  oriented simple
closed curves  on $\Sigma_g$ (with the blackboard framing) and
  $[\gamma]$ is its
homology class in $H_1(\Sigma_g,{\mathbb Z})={\mathbb Z}^{2g}$.
\end{theorem}

\begin{proof}
Parts (a) and (b) are consequences of a general result in \cite{przytycki2};
we include their proof for sake of completeness. \\
(a) Bring all skeins in the blackboard framing of the surface. A skein
$t^kL$, where $L$ is an oriented framed link in $\Sigma_g\times[0,1]$
is equivalent modulo  skein relations to a skein $t^{k+m}L'$ where
$L'$ is an oriented framed link such that
the projection of $L'$ onto the surface has no crossings, and
$m$ is the difference between the number of positive  and
 negative
crossings of the projection of $L$. Moreover, since
any embedded ball can be isotoped to a cylinder over a disk,
any skein $t^nL''$ that is equivalent to $t^kL$, with $L''$
 a framed link with no crossings, has the property that $n=k+m$.

\begin{figure}[ht]
\centering
\resizebox{.45\textwidth}{!}{\includegraphics{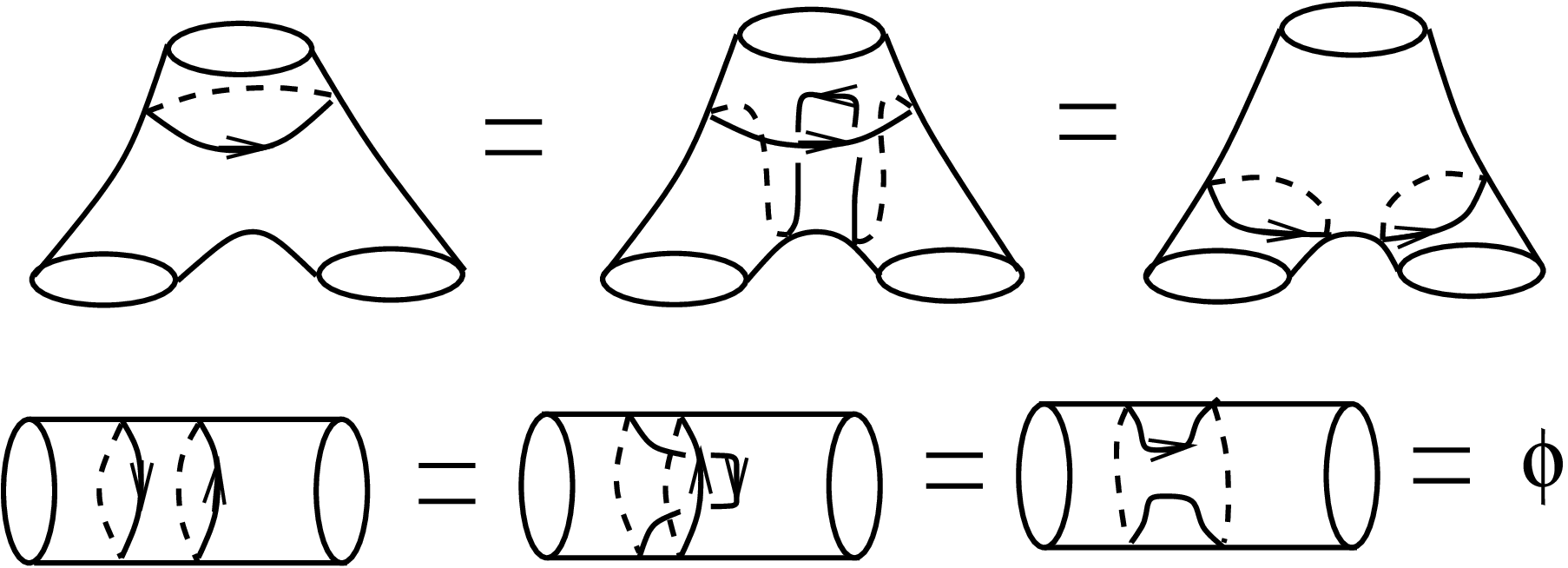}}
\caption{}
\label{nullhomologous}
\end{figure}

If $L$ is an oriented link with blackboard framing  whose projection onto
the surface has no crossings, and if it is null-homologous
in $H_1(\Sigma_g\times [0,1], {\mathbb Z})$, then $L$
is equivalent modulo skein relations to the empty skein. This
follows from the computations in Figure~\ref{nullhomologous}
since  $\Sigma_g$
can be cut into pairs of pants and annuli.

View $\Sigma_g$ as a sphere with $g$ punctured tori attached.
 Then $L$ is equivalent to a link $L'$ consisting of simple
closed curves on the tori, which therefore is of the form
\begin{eqnarray*}
(p_1,q_1)^{k_1}(p_2,q_2)^{k_2}\cdots (p_g,q_g)^{k_g},
\end{eqnarray*}
where $(p_j,q_j)$ denotes the curve of slope $p_j/q_j$ on the $j$th
torus. This last link is equivalent, modulo skein relations, to
\begin{eqnarray}\label{finalskein}
t^{\sum_jk_jp_jq_j}a_1^{k_1p_1}a_2^{k_2p_2}\cdots
a_g^{k_gp_g}b_1^{k_1q_1}b_2^{k_2q_2}
\cdots b_g^{k_gq_g}.
\end{eqnarray}
It is easy to check that if we change the link by a Reidemeister
move, then resolve all crossings, we obtain the same expression
(\ref{finalskein}). So the result  only depends on the link and
not on how it projects to $\Sigma_g$.
This proves (a).

Part (b)  is analogous to (a), given that a genus $g$ handlebody is
 the cylinder over a disk with $g$ punctures.
For (c) recall Corollary~\ref{groupalgheis}.
That the specified map is a linear isomorphism follows
from (a). It is straightforward to check
that the multiplication rule is the same.
\end{proof}

\begin{remark}
Explicitly, the map
\begin{eqnarray*}
t^ka_1^{m_1}a_2^{m_2}\cdots a_g^{m_g}b_1^{n_1}b_2^{n_2}\cdots b_g^{n_g}\mapsto\!
(m_1,m_2,\ldots, m_g,n_1,n_2,\ldots, n_g,k), \, m_j,n_j,k\in {\mathbb Z}
\end{eqnarray*}
defines an  algebra isomorphism between ${\mathcal L}(\Sigma_g\times
[0,1])$ and ${\mathbb C}[\heis({\mathbb Z}^g)]$.
\end{remark}

\begin{theorem}\label{heisenbergskeins2}
(a) The reduced linking number skein module
${\mathcal L}_N(\Sigma_g\times [0,1])$ is a finite dimensional
vector space with basis
\begin{eqnarray*}
a_1^{m_1}a_2^{m_2}\cdots a_g^{m_g}b_1^{n_1}b_2^{n_2}\cdots b_g^{n_g}, \quad
m_1,m_2,\ldots, m_g,n_1,n_2,\ldots, n_g\in {\mathbb Z}_N.
\end{eqnarray*}
(b)  The reduced linking number skein module ${\mathcal L}_N(H_g)$
is a finite dimensional  vector space with basis
\begin{eqnarray*}
a_1^{m_1}a_2^{m_2}\cdots a_g^{m_g},\quad m_1,m_2,\ldots, m_g\in {\mathbb Z}_N.
\end{eqnarray*}
Moreover, there is a linear isomorphism of
 ${\mathcal L}_N(H_g)$
and  $\spacetheta(\Sigma_g)$ given by
\begin{eqnarray*}
\gamma \rightarrow \theta_{[\gamma]},
\end{eqnarray*}
where $\gamma$ ranges among all oriented  simple closed curves in
$B^2_g$ with the blackboard framing and  $[\gamma]$ is the homology class of
$\gamma$ in $H_1(H_g,{\mathbb Z}_N^g)={\mathbb
  Z}_N^g$. \\
(c) The algebra isomorphism defined in Theorem~\ref{heisenbergskeins1}
factors to an algebra isomorphism of ${\mathcal
L}_N(\Sigma_g\times [0,1])$ and $L(\spacetheta(\Sigma_g))$, the algebra of
linear operators on the space of theta functions.
The isomorphism defined in (b)
intertwines the left action of ${\mathcal L}_N(
\Sigma_g\times [0,1])$ on ${\mathcal L}_N(H_g)$
and the Schr\"{o}dinger representation.
\end{theorem}

\begin{proof}
(a) By Theorem~\ref{heisenbergskeins1} we can identify ${\mathcal L}(
\Sigma_g\times [0,1])$ with ${\mathbb C}[\heis ({\mathbb Z}^g)]$. Setting $t=e^{\frac{i\pi}{N}}$
 and deleting any $N$ parallel copies of a link component
are precisely the relations by which we factor the Heisenberg group in
 Proposition~\ref{kernel}. The only question is whether factoring by
this additional relation  before applying the other skein relations
factors any further the skein module. However, we see that when
a curve is crossed by $N$ parallel copies of another curve, there is no
distinction between overcrossings and undercrossings. Hence if a link
contains $N$ parallel copies of a curve, we can move this curve so that
it is inside a cylinder $\Sigma_g\times [0,\epsilon]$ that does not
contain other link components and we can resolve all self-crossings
of this curve without introducing factors of $t$. Then we can delete the
curve without introducing new factoring relations.
 This proves (a).

For (b), notice that we factor ${\mathcal L}_N(\Sigma_g\times [0,1])$
to obtain  ${\mathcal L}_N(H_g)$ by the same relations by which
we factor ${\mathbb C}[\heis ({\mathbb Z}_N^g)]$ to obtain ${\mathcal H}_{N,g}({\bf L})$
in Section 3.

(c) An easy check shows that that the left action
of the skein algebra of the cylinder over the surface
on the skein module of the handlebody is the same as the one from
Propositions~\ref{weylquantization} and \ref{abstracttheta}.
\end{proof}

\begin{remark}
The isomorphism between the reduced skein module of the handlebody
and the space of theta functions is given explicitly by
\begin{eqnarray*}
a_1^{n_1}a_2^{n_2}\cdots a_g^{n_g}\mapsto \theta_{n_1,n_2,\ldots, n_g}^\Pi,
\quad \mbox{for all }n_1,n_2,\ldots, n_g\in {\mathbb Z}_N.
\end{eqnarray*}
\end{remark}

In view of Theorem~\ref{heisenbergskeins2} we endow ${\mathcal L}_N(H_g)$
with the Hilbert space structure of the space of theta functions.

Now we turn our attention to the discrete Fourier transform,
and translate in topological language  formula
(\ref{fouriertransform}).
Let  $h$ be an element of the mapping class
group of $\Sigma_g$.
The action of the mapping class group  on the finite
Heisenberg group from Section~\ref{sec:2} becomes the
action on skeins in $\Sigma_g\times[0,1]$ given by
\begin{eqnarray*}
\sigma\mapsto h(\sigma),
\end{eqnarray*}
where $h(\sigma)$ is obtained by replacing each framed curve of the skein
$\sigma$ by its image through the homeomorphism $h$.

 Consider $h_1$ and $h_2$ two
 homeomorphisms of
$\Sigma_g$ onto the boundary of the handlebody $H_g$ such that
$h_2=h\circ h_1$.
These homeomorphisms extend to embeddings
of ${\Sigma_g}\times [0,1]$ into $H_g$
which we denote by $h_1$ and $h_2$ as well.
The homeomorphisms $h_1$ and $h_2$ define the action
of ${\mathcal L}_N(\Sigma_g\times [0,1])$ on
${\mathcal L}_N(H_g)$  in two different ways,
i.e. they give two different constructions of the Schr\"{o}dinger
representations. By the Stone-von Neumann theorem, these are
unitary equivalent; they are related by the isomorphism $\rho(h)$.
We now give
$\rho(h)$ a topological definition. For this, let
us take a closer look at the lifting map $s_{\bf L}$ defined in (\ref{lift}).
First, it is standard to remark that one should only average
over $\exp({\bf L}+{\mathbb Z}E)/\exp({\mathbb Z}E)=\exp({\bf L})$,
hence
\begin{eqnarray*}
s_{\bf L}({\bf u}\mbox{ mod }\mbox{ker}(\pi_{\bf L}))=
\frac{1}{N^g}\sum_{u_1\in \exp({\bf L})}{\bf u}u_1.
\end{eqnarray*}

If ${\bf u}=u\in \heis({\mathbb Z}_N^g)$, then, as a skein,
 $u$ is of the form $\gamma^k$ where $\gamma$
is a framed oriented curve on $\Sigma_g=\partial H_g$ and  $k$ is
an integer. The equivalence class $\hat{u}=u\mbox{ mod }\mbox{ker}(\pi_{\bf L}(u))$  is
just this skein viewed as lying inside the handlebody; it consists
of $k$ parallel framed oriented curves in $H_g$.

On the other hand, as a skein, $u_1$ is of the form $b_1^{n_1}b_2^{n_2}
\ldots b_g^{n_g}$, and as such, the product $uu_1$ becomes, after
smoothing all crossings, another lift of the skein
$\hat{u}$ to the boundary obtained by lifting $\gamma$ to the boundary
and then taking $k$ parallel copies.
Such a lift is obtained by
pushing $\hat{u}$ inside a regular neighborhood of the boundary
and then viewing it as an element in ${\mathcal L}_N(\Sigma_g\times
[0,1])$. When $u_1$ ranges over all $\exp({\bf L})$ we obtain
all possible lifts of $\hat{u}$ to the boundary obtained by pushing
$\gamma $ to the boundary and then taking $k$ parallel copies.



\begin{theorem}\label{fourier}
For a skein of the form $\gamma^k$ in ${\mathcal L}_N(H_g)$,
where $\gamma $ is a curve in $H_g$ and $k$ a positive integer, consider
all  possible  liftings
to ${\mathcal L}_N(
\Sigma_g\times [0,1])$ using $h_1$, obtained by pushing the curve $\gamma$ to
the boundary and then taking $k$ parallel copies.
Take the average of these liftings
and map the average by $h_2$
to ${\mathcal L}_N(H_g) $.
 This defines a  linear endomorphism of $\widetilde{\mathcal L}_N(
H_g)$ which  is, up
to multiplication by a constant, the discrete Fourier transform $\rho(h)$.
\end{theorem}

\begin{proof}
The map defined this way intertwines
the Schr\"{o}dinger representations defined by $h_1$ and $h_2$, so
the theorem is a consequence of the Stone-von Neumann theorem.
\end{proof}

\noindent {\em Example:} We will exemplify this by showing how
the $S$-map on the torus  acts on the theta series
\begin{eqnarray*}
\theta_1^\Pi(z)=\sum_{n\in {\mathbb Z}}e^{2\pi i N\left[\frac{\Pi}{2}\left(\frac{1}{N}+n\right)^2 +
z\left(\frac{1}{N}+n\right)\right]}
\end{eqnarray*}
(in this case $\Pi$ is a just a complex number with positive imaginary
part).
This theta series is represented in the solid torus by the curve shown in
Figure~\ref{thetaS}.
\begin{figure}[h]
\centering
\scalebox{.30}{\includegraphics{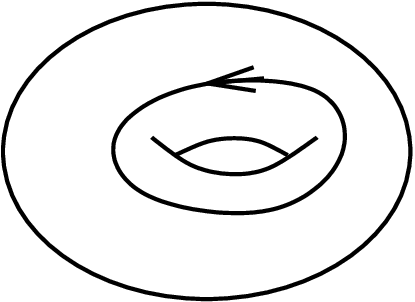}}
\caption{}
\label{thetaS}
\end{figure}
The  $N$ linearly independent
liftings of this curve to the boundary are shown in Figure~\ref{fourierS}.
 \begin{figure}[h]
\centering
\scalebox{.30}{\includegraphics{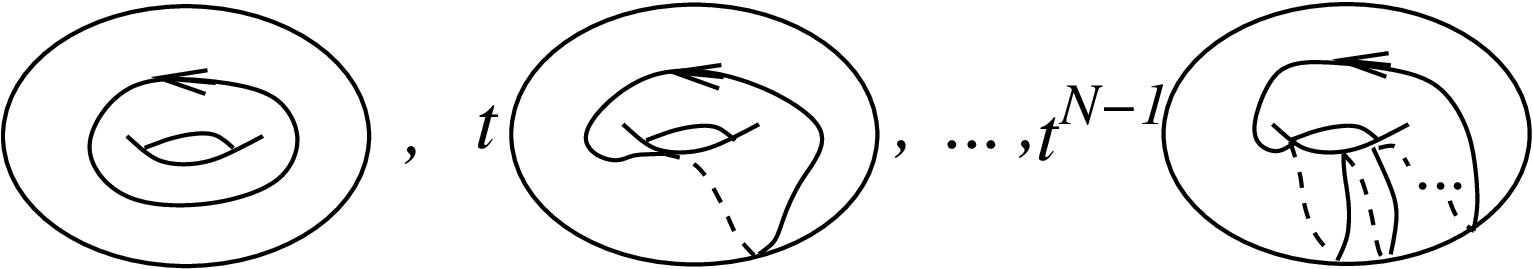}}
\caption{}
\label{fourierS}
\end{figure}

The $S$-map sends these to those in Figure~\ref{fourier2S},
which, after being pushed inside the solid torus, become the skeins
from Figure~\ref{fourier3S}.
\begin{figure}[h]
\centering
\scalebox{.30}{\includegraphics{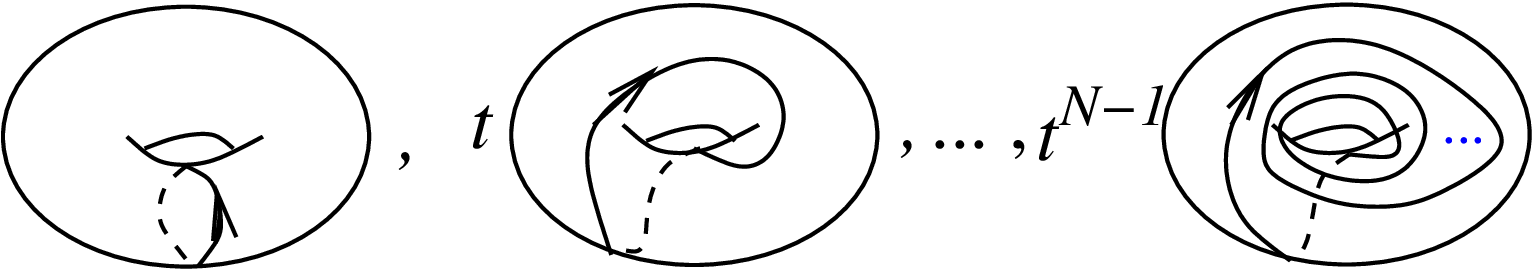}}
\caption{}
\label{fourier2S}
\end{figure}
\begin{figure}[h]
\centering
\scalebox{.30}{\includegraphics{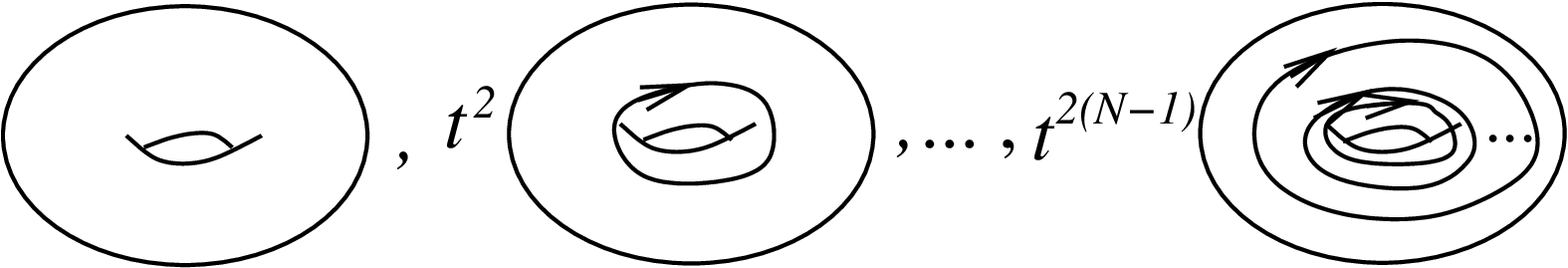}}
\caption{}
\label{fourier3S}
\end{figure}

Note that in each skein the arrow points  the opposite way as for
$\theta_1(z)$. Using the identity $\gamma^N=\emptyset$, we can replace
  $j$ parallel strands by $N-j$ parallel strands with opposite orientation.
Hence   these skeins are $t^{2j}\theta_{N-j}$, $j=1,\ldots, N$ (note also that
$\theta_{0}(z)=\theta_N(z)$). Taking the average we obtain
\begin{eqnarray*}
\rho(S)\theta_1(z)=\frac{1}{{N}}\sum_{j=0}^{N-1}
e^{\frac{2\pi i j}{N}}\theta_{N-j}(z)=
\frac{1}{{N}}\sum_{j=0}^{N-1}e^{-\frac{2\pi i j}{N}}\theta_j(z),
\end{eqnarray*}
which is, up to a multiplication by a constant, the standard
discrete Fourier transform of $\theta_1(z)$.

\section{The discrete Fourier transform as a skein}\label{sec:5}

As a consequence of Proposition~\ref{allspace},  $\rho(h)$
can be represented as an element in ${\mathbb C}[\heis({\mathbb Z}_N^g)]$.
Furthermore, Theorem~\ref{heisenbergskeins2} implies that $\rho(h)$
can be represented as left multiplication by a skein ${\mathcal F}(h)$ in
$\widetilde{\mathcal L}_t(\Sigma_g\times [0,1])$. The skein
${\mathcal F}(h)$ is unique up to a multiplication by a constant.
We wish to find an explicit formula for it.

Theorem~\ref{heisenbergskeins2} implies that the action of the group algebra of the
finite Heisenberg group can be represented as left multiplication by skeins.
Using this fact,  the exact Egorov identity (\ref{egorov}) translates to
\begin{eqnarray}\label{skeinegorov}
h(\sigma){\mathcal F}(h)={\mathcal F}(h)\sigma\mbox{ for all }\sigma\in {\mathcal L}_N(\Sigma_g\times [0,1])
\end{eqnarray}

By the Lickorish twist theorem (Chapter 9 in \cite{rolfsen}), every
homeomorphism of  $\Sigma_g$ is isotopic to a product of Dehn
twists along the $3g-1$ curves depicted in Figure~\ref{lickorish}. Recall
that a Dehn twist is the homemorphism obtained by cutting the surface
along the curve, applying a full rotation on one side, then gluing
back.

\begin{figure}[ht]
\centering
\scalebox{.28}{\includegraphics{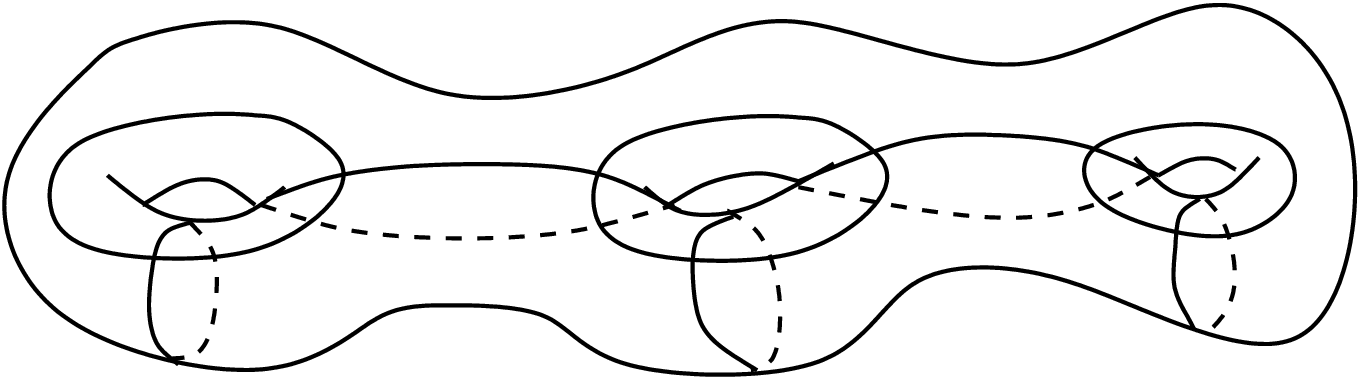}}
\caption{}
\label{lickorish}
\end{figure}

The curves from Figure~\ref{lickorish}  are nonseparating, and
any two can be mapped into one another by a homeomorphism of the surface.
Thus, to
understand ${\mathcal F}(h)$ in general it suffices to consider the case
$h=T$, the positive Dehn twist along the curve $b_1$ from Figure~\ref{homology}.
The word positive means that after we cut the surface along $b_1$ we perform
a full rotation of the part on the left in the direction of the arrow.
 Because
$T(\sigma)=\sigma$ for all skeins that do not contain curves that
intersect $b_1$, it follows that $\rho(T)$ commutes with all
such skeins. It also commutes with the multiples of $b_1$ (viewed as a skein
with the blackboard framing).
Hence $\rho(T)$ commutes with
all operators of the form $\exp(pP+qQ+kE)$ with $p_1$, the first
entry of $p$, equal to $0$. This implies that
\begin{eqnarray*}
\rho(T)=\sum_{j=0}^{N-1}c_j\exp(jQ_1).
\end{eqnarray*}
To determine the coefficients $c_j$, we write the exact Egorov identity
(\ref{egorov}) for $\exp(P_1)$. Since  $T\cdot \exp(P_1)=\exp(P_1+Q_1)$
this identity reads
\begin{eqnarray*}
\exp(P_1+Q_1)\sum_{j=0}^{N-1}c_j\exp(jQ_1)=\sum_{j=0}^{N-1}c_j\exp(jQ_1)\exp(P_1).
\end{eqnarray*}
We transform this further into
\begin{eqnarray*}
\sum_{j=0}^{N-1}c_je^{\frac{\pi i}{N}j}\exp[P_1+(j+1)Q_1]=
\sum_{j=0}^{N-1}c_je^{-\frac{\pi i}{N}j}\exp(P_1+jQ_1),
\end{eqnarray*}
or, taking into account that $\exp(P_1)=\exp(P_1+NQ_1)$,
\begin{eqnarray*}
\sum_{j=0}^{N-1}c_{j-1}e^{\frac{\pi i}{N}(j-1)}\exp(P_1+jQ_1)=
\sum_{j=0}^{N-1}c_je^{-\frac{\pi i}{N}j}\exp(P_1+jQ_1),
\end{eqnarray*}
where $c_{-1}=c_{N-1}$. It follows that $c_j=e^{\frac{\pi i}{N}(2j-1)}c_{j-1}$
for all $j$. Normalizing so that $\rho(T)$ is a unitary map
and $c_0>0$ we obtain $c_j=N^{-1/2}e^{\frac{\pi i}{N}j^2}$, and hence
\begin{eqnarray*}
{\mathcal F}(T)=N^{-1/2}\sum_{j=0}^{N-1}e^{\frac{\pi i}{N}j^2}\exp(jQ_1).
\end{eqnarray*}
Turning to the language of skein modules, and taking into account
that any Dehn twist is conjugate to the above twist by an element of the
mapping class group, we conclude that if $T$ is a positive
 Dehn twist along the
simple closed curve $\gamma $ on $\Sigma_g$, then
\begin{eqnarray*}
{\mathcal F}(T)=N^{-1/2}\sum_{j=0}^{N-1}t^{j^2}\gamma^j.
\end{eqnarray*}
This is the same as the skein
\begin{eqnarray*}
{\mathcal F}(T)=N^{-1/2}\sum_{j=0}^{N-1}(\gamma^+)^j
\end{eqnarray*}
where $\gamma^+$ is obtained by adding one full positive twist to
the framing of $\gamma $ (the twist is positive in the sense
that, as skeins, $\gamma^+=t\gamma$).

This skein has an interpretation in terms of surgery.
Consider the curve $\gamma^+\times\{1/2\}\subset \Sigma_g\times[0,1]$
with framing defined by the blackboard framing of $\gamma^+$ on $\Sigma_g$.
Take a solid torus which is a
regular neighborhood of the curve on whose boundary
the framing determines two simple closed curves.
Remove it from $\Sigma_g\times [0,1]$,
then glue it back in by a homeomorphism that identifies its meridian (the
curve that is null-homologous) to one of the curves determined by the
framing. This operation, called surgery, yields a manifold that
is homeomorphic to $\Sigma_g\times [0,1]$, such that the restriction of
the homeomorphism to $\Sigma_g\times \{0\}$ is the identity map,
and the restriction to $\Sigma_g\times \{1\}$ is the Dehn twist $T$.

The reduced linking number skein module of
the solid torus $H_1$ is,  by Theorem~\ref{heisenbergskeins2},
an $N$-dimensional vector space with basis $\emptyset,a_1,\ldots,a_1^{N-1}$.
Alternately, it is the vector space of 1-dimensional theta functions
with basis $\theta_0^\Pi (z), \theta_1^\Pi(z),\ldots, \theta_{N-1}^\Pi(z)$,
where $\Pi$ in this case is a complex number with positive imaginary part.
We introduce the element
\begin{eqnarray}\label{omega}
\Omega=N^{-1/2}\sum_{j=0}^{N-1}a_1^j=N^{-1/2}\sum_{j=0}^{N-1}\theta_j^\Pi(z)
\end{eqnarray}
in ${\mathcal L}_N(H_1)=\spacetheta(\Sigma_1)$.
As a diagram,  $\Omega$ is the skein depicted in Figure~\ref{omegafigure}
multiplied by $N^{-1/2}$.
\begin{figure}[ht]
\centering
\scalebox{.30}{\includegraphics{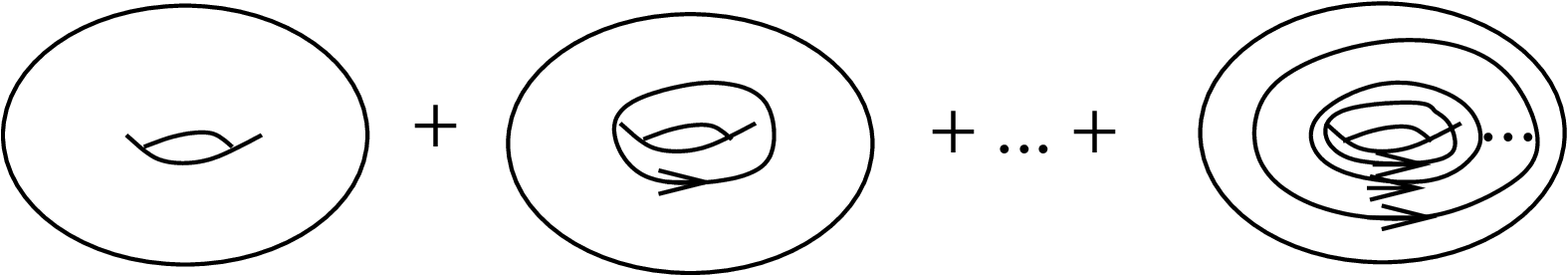}}
\caption{}
\label{omegafigure}
\end{figure}
If $S$ is the homemorphism on the torus
induced by the $90^\circ$ rotation of the plane when viewing the torus
as the quotient of the plane by the integer lattice, then
$\Omega =\rho(S)\emptyset$.  So $\Omega$ is the
(standard) discrete Fourier transform of $\theta_0^\Pi(z)$.

For an arbitrary framed link $L$ we denote by $\Omega(L)$ the skein
obtained by replacing each link component by $\Omega$. In
other words, $\Omega(L)$ is the sum of framed links obtained from $L$ by
replacing its components, in all possible ways, by $0,1,\ldots, N-1$
parallel copies. The skein $\Omega$ is called the coloring of $L$ by
$\Omega$.


\begin{proposition}\label{omegaproperties}
a) The skein $\Omega(L)$ is independent of the orientations of the components
of $L$.\\
b) The skein relation from Figure~\ref{omegalink} holds, where the $n$
parallel strands point in the same direction.
\end{proposition}

\begin{figure}[ht]
\centering
\scalebox{.27}{\includegraphics{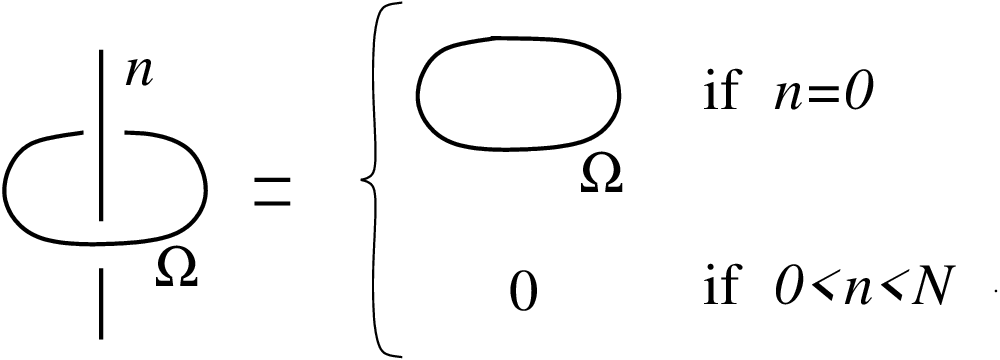}}
\caption{}
\label{omegalink}
\end{figure}

\begin{proof}
a)
The computation in Figure~\ref{nullhomologous} implies that if we switch
the orientation on the $j$ parallel curves that represent $\theta_j^\Pi(z)$
we obtain $\theta_{N-j}^\Pi(z)$. Hence by changing the orientation on all  curves
that make up  $\Omega$ we obtain the skein
\begin{eqnarray*}
N^{-1/2}[\theta_0^\Pi(z)+\theta_{N-1}^\Pi(z)+\theta_{N-2}^\Pi(z)+\cdots +
\theta_1^\Pi(z)],
\end{eqnarray*}
which is, again, $\Omega$.

b) When $n=0$ there is nothing to prove. If $n\neq 0$, then by
resolving all crossings in the diagram we obtain $n$ vertical parallel
strands with the coefficient
\begin{eqnarray*}
N^{-1/2}\sum_{j=0}^{N-1}t^{\pm 2nj}=N^{-1/2}\cdot \frac{t^{2Nn}-1}{t^{2n}-1}.
\end{eqnarray*}
where the signs in the exponents are either all positive, or all negative.
Since $t^2$ is a primitive $N$th root of unity, this is equal to zero.
Hence the conclusion.
\end{proof}

Up to this point we have proved the following result:

\begin{lemma}
For a Dehn twist $T$,  ${\mathcal F}(T)$ is the skein obtained
by coloring the surgery framed curve $\gamma^+$ of $T$ by  $\Omega$.
\end{lemma}

Since by the Lickorish twist theorem every element $h$ of the mapping
class group is a product of twists, 
we obtain the following skein
theoretic description of the discrete Fourier transform induced by the
map $h$.

\begin{proposition}\label{skeinfourier}
Let $h$ be an element of the mapping class group of $\Sigma_g$ obtained
as a composition of Dehn twists $h=T_1T_2\cdots T_n$.
Express each Dehn twist $T_j$ by surgery
on a curve $\gamma_j$ as above, and consider the link
 $L_h=\gamma_1\cup \gamma_2\cup
\cdots \cup \gamma_n$ which expresses $h$ as
 surgery on the framed link $L_h$ in $\Sigma_g\times [0,1]$.
Then  the discrete Fourier transform $\rho(h):\widetilde{\mathcal L}_t(H_g)
\rightarrow \widetilde{\mathcal L}_t(H_g)$ is given by
\begin{eqnarray*}
\rho(h)\beta=\Omega(L_h)\beta.
\end{eqnarray*}
\end{proposition}

\section{The Egorov identity and handle slides}\label{sec:6}

Next, we give the Egorov identity a topological  interpretation in
terms of handle slides. For this we look at its skein theoretical version
(\ref{skeinegorov}). We start again with an example on the torus.

\medskip

\noindent {\em Example:} For the positive twist $T$ and the operator
represented by the curve $(1,0)$ the exact Egorov identity reads
\begin{eqnarray*}
\rho(T)(1,0)=(1,1)\rho(T),
\end{eqnarray*}
which is shown in Figure~\ref{exegorov}.
\begin{figure}[ht]
\centering
\scalebox{.28}{\includegraphics{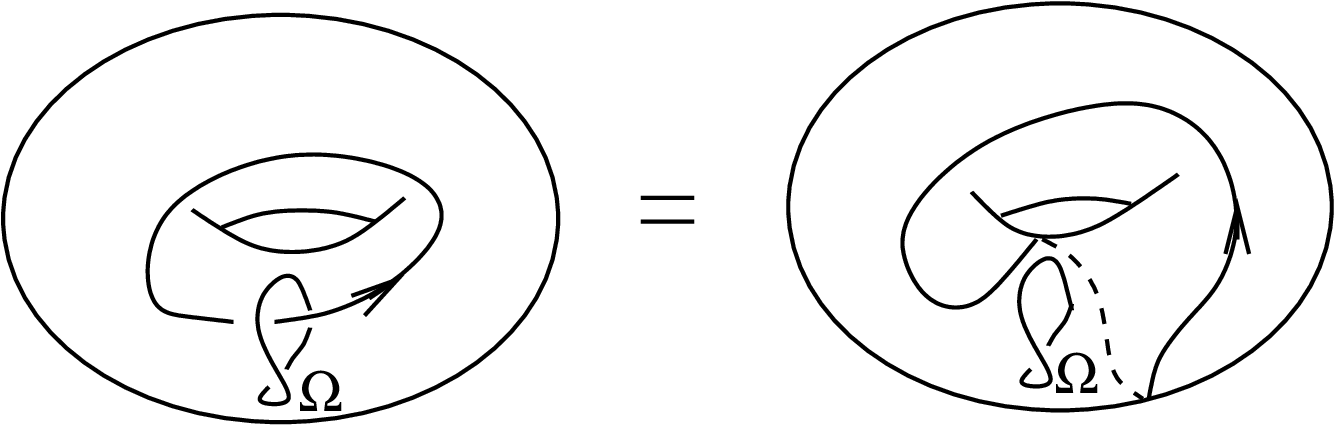}}
\caption{}
\label{exegorov}
\end{figure}
The diagram on the right is the same as the one in
 Figure~\ref{egorovslide}.
\begin{figure}[ht]
\centering
\scalebox{.28}{\includegraphics{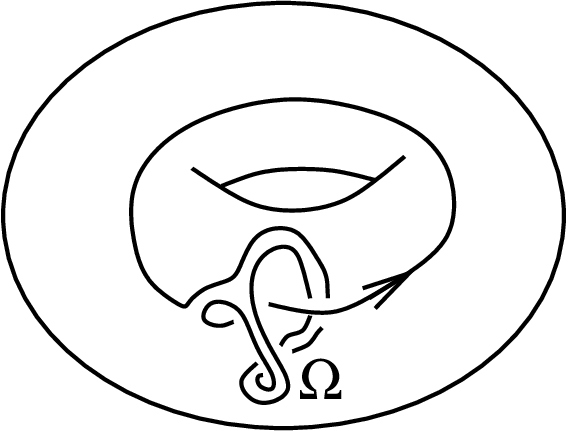}}
\caption{}
\label{egorovslide}
\end{figure}
As such, the curve $(1,1)$ is obtained by sliding the curve $(1,0)$
along the surgery curve of the positive twist. Here is the
detailed description of the
operation of sliding a framed knot along another using a Kirby
band-sum move.

The {\em slide} of
a framed knot $K_0$ along the framed knot $K$,
denoted by $K_0\#K$, is obtained
as follows. Let $K_1$ be a copy of $K$ obtained by
pushing $K$ in the direction of its framing.
Take an embedded $[0,1]^3$ that is disjoint from
$K,K_0$, and $K_1$ except for the opposite faces
$F_i=[0,1]^2\times \{i\}$, $i=0,1$ and  which are embedded in
$\partial K_0$
respectively $\partial K_1$. $F_i$ is embedded
in the annulus $K_i$ such that $[0,1]\times \{j\}\times\{i\}$ is
embedded in $\partial K_i$. Delete from $K_0\cup K_1$
the faces $F_i$ and add the faces $\{j\}\times[0,1]\times [0,1]$.
The framed knot obtained this way is $K_0\#K$.
Saying it less rigorously but more intuitively, we cut
the knots $K_0$ and $K_1$ and join together the two open strands
by pulling them along the sides of an embedded rectangle (band) which
does not intersect the knots.    Figure~\ref{slide}
shows the slide of a trefoil knot over a figure-eight knot, both
with the blackboard framing. When the knots are oriented, we perform the
slide so that the orientations match.
One should point out that there are many ways in which one
can slide one knot along the other, since the band that
connects the two knots is not unique.

\begin{figure}[ht]
\centering
\scalebox{.25}{\includegraphics{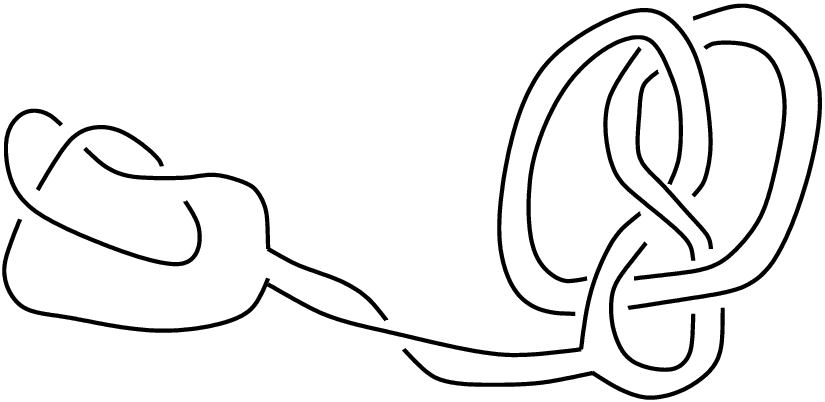}}
\caption{}
\label{slide}
\end{figure}

For a closed curve $\alpha$ in $\Sigma_g=\Sigma_g\times \{0\}$,
the curve $h(\alpha)$ is obtained from $\alpha$ by slides over
the components of the surgery link of $h$. Indeed, if $h$ is the twist along
the curve $\gamma$, with surgery
curve $\gamma^+$, and if $\alpha$ and $\gamma$ intersect on
$\Sigma_g$ at only one point, then $h(\alpha)=\alpha\#\gamma^+$. 
 If the algebraic
intersecton number of $\alpha$ and $\gamma$ is $\pm k$, then
$h(\alpha)$ is obtained from $\alpha$ by performing $k$ consecutive
slides along $\gamma^+$. The general case follows from the fact that
$h$ is a product of twists.


It follows that the exact Egorov identity is a particular case of
slides of  framed knots along components of the  surgery link. In fact, the
exact Egorov identity covers all cases of slides of one knot along
another knot colored by $\Omega$, and we have

\begin{theorem}\label{handleslide}
Let $M$ be a $3$-manifold, $\sigma$ a
skein in ${\mathcal L}_N(M)$ and $K_0$ and $K$ two oriented
framed knots in $M$ disjoint from $\sigma$. Then, in
${\mathcal L}_N(M)$, one has
\begin{eqnarray*}
\sigma\cup K_0\cup \Omega(K)=\sigma\cup (K_0\#K)\cup \Omega(K),
\end{eqnarray*}
however one does the band-sum $K_0\#K$.
\end{theorem}

\begin{remark}
The knots from the statement of the theorem should be understood as
representing elements in ${\mathcal L}_N(M)$.
\end{remark}

\begin{proof}
Isotope $K_0$ along the embedded $[0,1]^3$ that defines $K_0\#K$ to a knot
$K_0'$ that  intersects $K$. There is an embedded punctured
torus $\Sigma_{1,1}$ in $M$,
disjoint from $\sigma$, which contains $K_0'\cup K$ on
its boundary, as shown in Figure~\ref{puncturedtorus} a). In fact,
by looking at  a neighborhood of this torus, we can find an embedded
$\Sigma_{1,1}\times [0,1]$ such that $K_0'\cup K\subset \Sigma_{1,1}\times
\{0\}$. The boundary of this cylinder is a genus $2$ surface $\Sigma_2$,
and $K_0'$ and $K_1$ lie in a punctured torus of this surface and intersect
at exactly one point. By pushing off $K_0'$ to a knot isotopic to $K_0$ (which
we identify with $K_0$), we see that we can  place $K_0$ and $K$ in an embedded
$\Sigma_2\times [0,1]$ such that $K_0\in \Sigma_2\times \{0\}$
and $K\in \Sigma_2\times \{1/2\}$.

\begin{figure}[ht]
\centering
\scalebox{.20}{\includegraphics{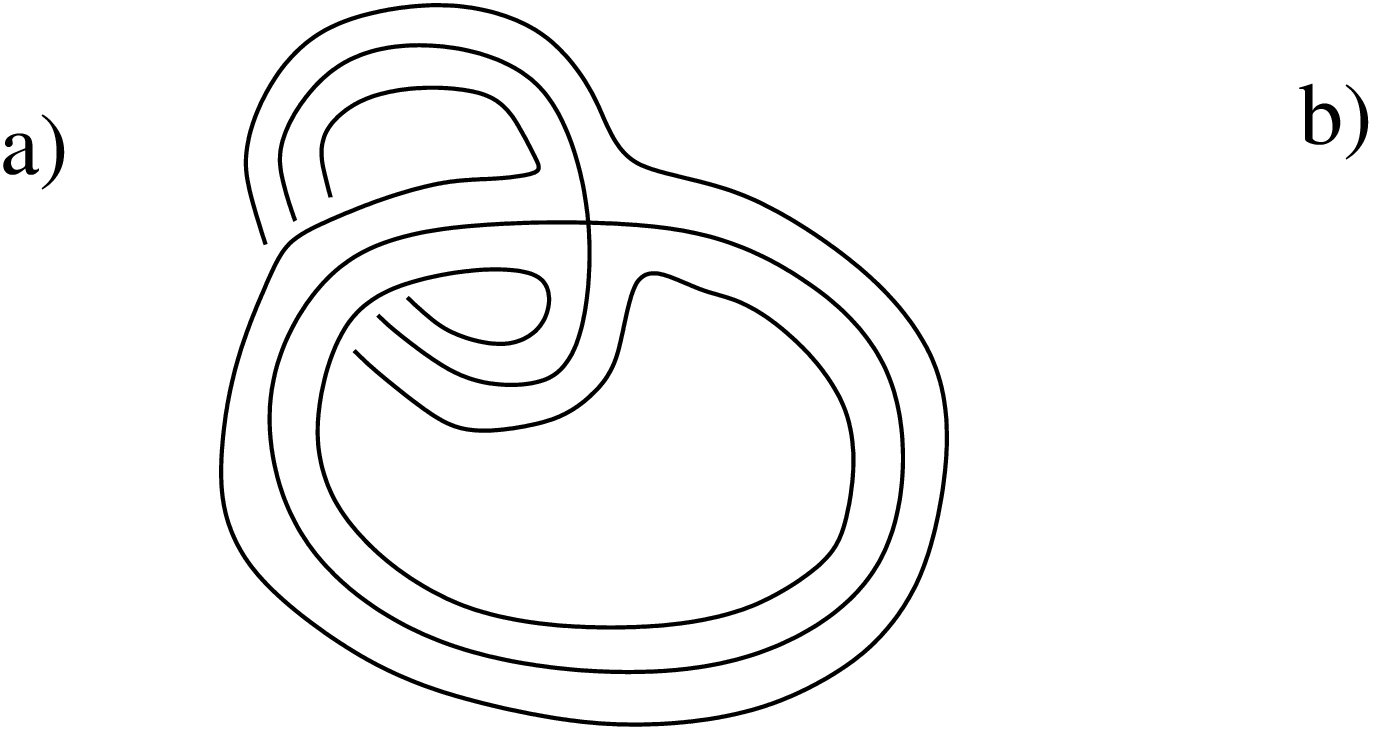}}
\scalebox{.25}{\includegraphics{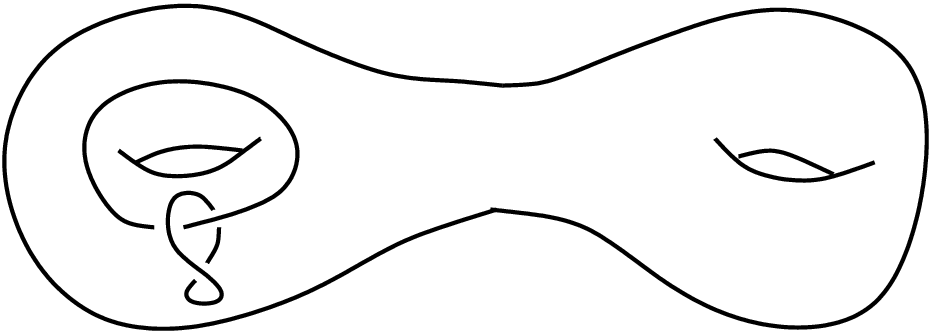}}
\caption{}
\label{puncturedtorus}
\end{figure}

By performing a twist in $\Sigma_{1,1}\times [0,1]$ we can change
the framing of $K$ in such a way that $K_0$ and $K$ look inside
$\Sigma_2\times [0,1]$ like  in Figure~\ref{puncturedtorus} b).
Then
$K_0$ is mapped to $K_0\#K$ in $\Sigma_2\times \{1\}$ by the Dehn twist
of $\Sigma_2$ with surgery diagram $K$.
Hence the equality
\begin{eqnarray*}
K_0\cup \Omega(K)=(K_0\#K)\cup \Omega(K)
\end{eqnarray*}
in $\Sigma_2\times [0,1]$ is just the exact Egorov identity, which we
know is true . By embedding $\Sigma_2\times [0,1]$ in $\Sigma_{1,1}\times
[0,1]$ we conclude that
this equality holds in $\Sigma_{1,1}\times [0,1]$. By applying the inverse
of the twist, embedding $\Sigma_{1,1}\times [0,1]$ in $M$, and adding
$\sigma$,  we conclude that the identity from the statement holds as well.
 \end{proof}

The operation of sliding one knot along another is related
to the surgery description of 3-manifolds
(see \cite{rolfsen}). We  recall the
basic facts.

We use the standard notation
$B^n$ for  an $n$-dimensional (unit) ball and
$S^n$ for the $n$-dimensional sphere.
Every oriented, closed, $3$-dimensional
manifold is the boundary of a $4$-dimensional manifold
obtained by adding $2$-handles ${B}^2\times {B}^2$  to $B^4$ along the
solid tori $B^2\times S^1$ \cite{lickorish}.
 On the boundary $S^3$ of $B^4$, when adding a handle we remove
a solid torus from $S^3$ (the one identified with $B^2\times S^1$)
and glue back the solid torus $S^1\times B^2$. The curve
$\{1\}\times S^1$ in the solid torus $B^2\times S^1$ that
was removed becomes the null-homologous curve on
the boundary of $S^1\times B^2$.

This procedure of constructing
$3$-manifolds is called Dehn surgery with integer
coefficients. The curve $\{1\}\times S^1$ together with the
core of  $B^2\times S^1$ bound an embedded annulus
which defines a framed link component in $S^3$. So the information
for Dehn surgery
with integer coefficients is encoded in a framed link
in $S^3$.

If $K_0$ is a knot inside a 3-dimensional manifold $M$
obtained by surgery on $S^3$ and if the framed knot $K$ is a
component of the surgery link, then $K_0\#K$ is the slide of
$K_0$ over the $2$-handle corresponding to $K$.
Indeed, when we slide $K_0$ along the handle we push one arc close
to $K$, then move it to the other side of the handle by pushing
it through the meridinal disk of the surgery solid torus. The
meridian of this solid torus is parallel to the knot $K$ (when viewed in
$S^3)$, so by sliding $K_0$ across the handle we obtain $K_0\#K$.
In particular, the operation of sliding one 2-handle over another
corresponds to sliding one link component of the surgery link along
another.

In conclusion, we can say that the Egorov identity allows handle-slides
along surgery link components colored by $\Omega$. We will make use of
 this fact
in Section~7.

\section{The topological quantum field theory
 associated to theta functions}\label{sec:7}

Theorem~\ref{handleslide} has two direct consequences:
\begin{itemize}
\item the definition
of a topological invariant for closed $3$-dimensional manifolds,
\item the existence of an isomorphism between the reduced linking number
skein modules of $3$-dimensional manifolds with homeomorphic boundaries.
\end{itemize}

We have seen  in the previous section that handle-slides correspond to changing
the presentation of a $3$-dimensional manifold as surgery on a framed link.
Kirby's theorem \cite{kirby1} states that two framed link diagrams
represent the same $3$-dimensional manifold if they can be transformed
into one another by a sequence of isotopies, handle slides, and
 additions/deletions
of the trivial link components $U_+$ and $U_{-}$  described
in Figure~\ref{trivialhandle}. A  trivial link component
corresponds to adding a  $2$-handle to $B^4$
in a trivial way, and on the boundary, to taking
the connected sum of the original $3$-dimensional manifold and  $S^3$.

\begin{figure}[ht]
\centering
\scalebox{.30}{\includegraphics{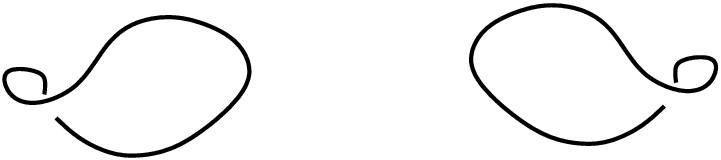}}
\caption{}
\label{trivialhandle}
\end{figure}

Theorem~\ref{handleslide} implies that, given a framed link $L$
in $S^3$, the element
$\Omega(L)\in {\mathcal L}_N(S^3)={\mathbb C}$ is an invariant of
the $3$-dimensional manifold obtained by performing surgery on
$L$, modulo addition and subtraction of trivial $2$-handles. This ambiguity
can be removed by using the linking matrix of $L$ as follows.

Recall that the  linking matrix of an oriented framed link $L$ has the
$(i,j)$ entry equal to the linking number of the $i$th and $j$th
components for $i\neq j$ and the $(i,i)$ entry equal to the
writhe of the $i$th component,
namely to the linking number of the $i$th component with a push-out of
this component in the direction of the framing.
The signature $\mbox{sign}(L)$ of the linking matrix does not depend
on the orientations of the components of $L$, and  is equal to
the signature $\mbox{sign}(W)$ of the 4-dimensional manifold $W$ obtained
by adding 2-handles to $B^4$ as specified by $L$. Here
 the signature of $\mbox{sign}(W)$ is the signature of the intersection
form in $H_2(W,{\mathbb R})$.

When adding a trivial handle via $U_+$ respectively $U_-$, the signature
of the linking matrix, and hence of the 4-dimensional manifold, changes
by $+1$ respectively $-1$.

\begin{proposition}\label{omegatrivial}
In any $3$-dimensional manifold, the following equalities hold
\begin{eqnarray*}
\Omega(U_+)=e^{\frac{\pi i}{4}}\emptyset,\quad \Omega(U_-)=
e^{-\frac{\pi i}{4}}\emptyset.
\end{eqnarray*}
Consequently $\Omega(U_+)\cup\Omega(U_-)=\emptyset$.
\end{proposition}

\begin{proof}
We have
\begin{eqnarray*}
\Omega(U_+)=N^{-1/2}\sum_{j=0}^{N-1}t^{j^2}\emptyset.
\end{eqnarray*}
Because $N$ is even,
\begin{eqnarray*}
\sum_{j=0}^{N-1}t^{j^2}=\sum_{j=0}^{N-1}e^{\frac{\pi i}{N}j^2}=\sum_{j=0}^{N-1}
e^{\frac{\pi i}{N}(N+j)^2}=\sum_{j=N}^{2N-1}t^{j^2}.
\end{eqnarray*}
Hence
\begin{eqnarray*}
\sum_{j=0}^{N-1}t^{j^2}=\frac{1}{2}\sum_{j=0}^{2N-1}e^{\frac{2\pi i}{2N}j^2}.
\end{eqnarray*}
The last expression is a Gauss sum, which is
equal to $e^{\frac{\pi i}{4}}N^{1/2}$ (see
\cite{lang} page 87). This proves the first formula.

On the other hand,
\begin{eqnarray*}
\Omega(U_-)=N^{-1/2}\sum_{j=0}^{N-1}e^{-\frac{\pi i}{N}j^2}\emptyset
\end{eqnarray*}
which is the complex conjugate of $\Omega(U_+)$. Hence, the
second formula.
\end{proof}

\begin{theorem}\label{manifoldinvariant}
Given a closed, oriented, $3$-dimensional manifold $M$ obtained
as surgery on the framed link $L$ in $S^3$, the number
\begin{eqnarray*}
Z(M)=e^{-\frac{\pi i}{4}\mbox{sign}(L)}\Omega(L)
\end{eqnarray*}
is a topological invariant of the manifold $M$.
\end{theorem}

\begin{proof}
Using Proposition~\ref{omegatrivial} we can rewrite
\begin{eqnarray*}
Z(M)=\Omega(U_+)^{-b_+}\Omega(U_-)^{-b_-}\Omega(L)
\end{eqnarray*}
where $b_+$ and $b_-$ are the number of positive, respectively
negative eigenvalues of the linking matrix. This quantity is
invariant under  addition of trivial handles, and also under
handleslides, by Theorem~\ref{handleslide}, so it is a topological
invariant of $M$.
\end{proof}

\begin{remark}
This is the Murakami-Ohtsuki-Okada invariant \cite{moo}! Here we
derived its existence directly from the theory of theta functions.
\end{remark}

The second application
of the exact Egorov identity is the construction of
 a Sikora isomorphism, which identifies
the reduced linking number skein modules of two  manifolds with
homeomorphic boundaries.
Let us point out that such an isomorphism was constructed for reduced
Kauffman bracket skein modules in \cite{sikora}.

\begin{theorem}\label{sikora}
Let $M_1$ and $M_2$ be two $3$-dimensional manifold with
homeomorphic boundaries. Then
\begin{eqnarray*}
{\mathcal L}_N(M_1)\cong {\mathcal L}_N(M_2).
\end{eqnarray*}
\end{theorem}

\begin{proof}
Because the manifolds $M_1$ and $M_2$ have homeomorphic boundaries,
there is a framed  link $L_1\subset M_1$  such that
$M_2$ is obtained by performing surgery on $L_1$ in $M_1$. Let $N_1$
be a regular neighborhood of $L_1$ in $M$, which is the union of
several solid tori, and let  $N_2$ be the union of the
surgery tori in $M_2$. The cores of these tori form a framed link
$L_2\subset M_2$, and $M_1$ is obtained by performing surgery on $L_2$
in $M_2$. Every skein in $M_1$, respectively $M_2$ can be isotoped to
one that misses $N_1$, respectively $N_2$. The homeomorphism $M_1\backslash
N_1\cong M_2\backslash N_2$ yields an isomorphism
\begin{eqnarray*}
\phi:{\mathcal L}_N(M_1\backslash N_1)\rightarrow
{\mathcal L}_N(M_1\backslash N_1).
\end{eqnarray*}
However, this does not induce a well defined map between
${\mathcal L}_N(M_1)$ and ${\mathcal L}_N(M_2)$ because
a skein can be pushed through the $N_i$'s. To make it
well defined, the skein should not change when pushed through these
regular neighborhoods. We use Theorem~\ref{handleslide}
and define
\begin{eqnarray*}
&&F_1:{\mathcal L}_N(M_1)\rightarrow
{\mathcal L}_N(M_1),\quad
F_1(\sigma)=\phi(\sigma)\cup \Omega(L_1)\\
&&F_2:{\mathcal L}_N(M_2)\rightarrow
{\mathcal L}_N(M_1),
\quad
F_2(\sigma)=\phi^{-1}(\sigma)\cup \Omega(L_2).
\end{eqnarray*}
By Proposition~\ref{omegaproperties} b) we have
\begin{eqnarray*}
\Omega(L_1)\cup \Omega(\phi^{-1}(L_2))=
\emptyset \in\widetilde{\mathcal L}_t(M_1),
\end{eqnarray*}
since each of the components of $\phi^{-1}(L_2)$ is a meridian in
the surgery torus, hence it surrounds exactly once the corresponding
component in $L_1$. This implies that $F_2\circ F_1=Id$. A similar
argument shows that $F_1\circ F_2=Id$.
\end{proof}

One should note that the Sikora isomorphism depends on the surgery diagram.
Now it is easy to describe the reduced linking number skein module
of any manifold.

\begin{proposition}\label{vectorspace}
For every oriented $3$-dimensional manifold $M$ having the boundary
components $\Sigma_{g_i}$, $i=1,2,\ldots, n$, one has
\begin{eqnarray*}
{\mathcal L}_N(M)\cong \bigotimes_{i=1}^n{\mathbb C}^{Ng_i}.
\end{eqnarray*}
\end{proposition}

\begin{proof}
If $M$ has no boundary  then ${\mathcal L}_N(M)=
{\mathcal L}_N(S^3)={\mathbb C}$, and if $M$ is bounded by a
sphere, then ${\mathcal L}_N(M)={\mathcal L}_N(B^3)=
{\mathbb C}$, where $B^3$ is the $3$-dimensional ball.
If $M$ has one genus $g$ boundary component with $g\geq 1$, then
${\mathcal L}_N(M)={\mathcal L}_N(H_g)={\mathbb C}^{Ng}$
by Theorem~\ref{heisenbergskeins2}.

To tackle the case of  more boundary components we need the
following result:

\begin{lemma}
Given two oriented $3$-dimensional manifolds $M_1$
and $M_2$,  let $M_1\#M_2$ be their connected sum. The map
 \begin{eqnarray*}
{\mathcal L}_N(M_1)\otimes
{\mathcal L}_N(M_2)\rightarrow {\mathcal L}_N(M_1\#M_2)
\end{eqnarray*}
defined by $(\sigma,\sigma')\mapsto \sigma\cup \sigma'$ is an
isomorphism.
\end{lemma}

\begin{proof}
In $M_1\#M_2$, the manifolds $M_1$ and $M_2$ are separated by a $2$-dimensional
sphere $S^2_{sep}$.  Every skein in $M_1\#M_2$ can be written as
$\sum_{j=0}^{N-1}\sigma_j$, where each $\sigma_j$ intersects $S^2_{sep}$
in $j$ strands pointing in the same direction. A trivial skein
colored by $\Omega$ is equal to the empty link. But when we slide
it over $S^2_{sep}$ it turns $\sum_{j=0}^{N-1}\sigma_j$ into $\sigma_0$. 
This shows that the map from the statement is onto.

On the other
hand, the reduced linking number skein module of a regular
neighborhood of $S^2_{sep}$ is ${\mathbb C}$ since every skein can be
resolved to the empty link. This means that, in $M_1\#M_2$,
 if a skein that lies entirely in  $M_1$ can
be isotoped to a skein that lies entirely in  $M_2$, then this skein is
a scalar multiple of the empty skein.
So if $\sigma_1\cup\sigma_1'=\sigma_2\cup \sigma_2'$, then
$\sigma_1=\sigma_2$ in ${\mathcal L}_N(M_1)$ and
$\sigma_1'=\sigma_2'$ in ${\mathcal L}_N(M_2)$. Hence the map
is one-to-one, and we are done.
\end{proof}

Returning to the theorem, an oriented $3$-manifold with
$n$ boundary components
can be obtained as surgery on a connected sum of $n$ handlebodies.
The conclusion  follows by applying the lemma.
\end{proof}

If $M$ is a $3$-manifold without boundary, then Theorem~\ref{sikora} shows that
\begin{eqnarray*}
{\mathcal L}_N(M)\cong {\mathcal L}_N(S^3)= {\mathbb C}.
\end{eqnarray*}
If we describe $M$ as surgery on a framed link $L$ with signature zero, which
is always possible by adding trivial link components with framing
$\pm 1$, then the Sikora isomorphism maps the empty link
in $M$ to  the vector
\begin{eqnarray*}
Z(M)=\Omega(L)\in {\mathcal L}_N(S^3)={\mathbb C}.
\end{eqnarray*}

More generally, $M$ can be endowed with a framing defined by the signature
of the $4$-dimensional manifold $W$ that it bounds constructed as
explained before. If $L$ is the surgery link that gives rise to $M$
and $W$, then to the framed manifold $(M,\mbox{sign}(W))=(M,m)$ we can
associate the invariant
\begin{eqnarray*}
Z(M,m)=\Omega(L)\in {\mathcal L}_N(S^3).
\end{eqnarray*}
The Sikora isomorphism associated to $L$
identifies this invariant with the empty link
in $M$.

All these can be generalized to manifolds with boundary. A $3$-dimensional
manifold $M$ with boundary  can be obtained by performing surgery on
a framed link $L$ in the  complement $M_0$ of $n$ handlebodies
embedded in $S^3$, $n\geq 1$. Endow  $M$ with
 a framing by filling in
the missing handlebodies in $S^3$ and constructing the $4$-dimensional
manifold $W$ with the surgery instructions from $L$.
To the  manifold $(M,\mbox{sign}(W))=(M,m)$ we can associate
the skein $\emptyset\in
{\mathcal L}_N(M)$. A Sikora isomorphism allows us to identify
this vector with $\Omega(L)\in {\mathcal L}_N(M_0)$. Another
Sikora isomorphism allows us to identify ${\mathcal L}_t(M_0)$
with the reduced linking number
skein module of the connected sum of handlebodies, namely with
$\otimes_{i=1}^n{\mathbb C}^{Ng_i}$, where $g_i$, $i=1,2,\ldots, n$ are the
genera of  the boundary components
of $M$.
Via Proposition~\ref{vectorspace}, $\Omega(L)$ can be identified
 with a vector
\begin{eqnarray*}
Z(M,m)\in  \bigotimes_{i=1}^n{\mathbb C}^{Ng_i}.
\end{eqnarray*}

There is another way to see this identification of the linking number
skein module with $\otimes_{i=1}^n{\mathbb C}^{Ng_i}$, done in the spirit
of \cite{reshetikhinturaev}. In this alternative formalism, $\Omega(L)$
acts as a linear functional on
\begin{eqnarray*}
\Omega(L):\bigotimes_{i=1}^n{L}_N(H_{g_i})\rightarrow {\mathbb C},
\end{eqnarray*}
by gluing handlebodies to $M_0$ as to obtain $S^3$. In this setting,
the formalism developed in \cite{turaev} Chapter IV applies to show that
the vector
\begin{eqnarray*}
Z(M,m)\in \left(\oplus_{i=1}^N{\mathbb C}^{Ng_i}\right)^*=\oplus_{i=1}^n{\mathbb C}^{Ng_i}
\end{eqnarray*}
is well defined once the framing $m$ is fixed. Let us point out
that this can be modeled with the quantum group of abelian Chern-Simons theory
in the framework of \cite{gelcahamilton}.

The construction fits Atiyah's formalism of a topological quantum field
theory (TQFT) \cite{atiyah} with anomaly \cite{turaev}. In this formalism
\begin{itemize}
\item to each surface $\Sigma=\Sigma_{g_1}\cup\Sigma_{g_2}\cup\cdots \cup
\Sigma_{g_n}$ we associate the vector space
\begin{eqnarray*}
V(\Sigma)= \bigotimes_{i=1}^n{\mathbb C}^{Ng_i}
\end{eqnarray*}
which is isomorphic to the reduced linking number skein module of any
$3$-dimensional manifold that $\Sigma$ bounds.
\item to each framed $3$-dimensional manifold $(M,m)$ we associate
the empty link in ${\mathcal L}_N(M)$.
As a vector in $V(\partial M)$, this is $Z(M,m)$.
\end{itemize}

Atiyah's axioms are easy to check. {\em Functoriality} is obvious.
The fact that $Z$ is {\em involutory} namely that $Z(\Sigma^*)=
Z(\Sigma)^*$ where $\Sigma^*$ denotes $\Sigma$ with opposite orientation
follows by gluing a manifold $M$ bounded by $\Sigma$ to a
manifold $M^*$ bounded by $\Sigma^*$ and using the standard pairing
\begin{eqnarray*}
{\mathcal L}_N(M)\times {\mathcal L}_N(M^*)\rightarrow
{\mathcal L}_N(M\cup M^*)={\mathbb C}.
\end{eqnarray*}
Let us check the {\em multiplicativity} of $Z$ for disjoint union.
 If $\Sigma$ and
$\Sigma'$ are two surfaces, we can consider disjoint
$3$-dimensional manifolds $M$ and $M'$ such that $\partial M=\Sigma$
and $\partial M'=\Sigma'$. Then
\begin{eqnarray*}
Z(\Sigma\cup \Sigma')={\mathcal L}_N(M\cup M')=
{\mathcal L}_N(M)\otimes{\mathcal L}_N(M')=
Z(\Sigma)\otimes Z(\Sigma').
\end{eqnarray*}
Also $\emptyset \in {\mathcal L}_N(M\cup M')$ equals
$\emptyset \otimes \emptyset \in {\mathcal L}_N(M)\otimes
{\mathcal L}_N(M')$. If we now endow $M$ and $M'$ with framings
$m$, respectively $m'$, then
\begin{eqnarray*}
Z(M\cup M',m+m')=Z(M,m)\otimes Z(M',m'),
\end{eqnarray*}
since the Sikora isomorphism acts separately on the skein modules of $M$
and $M'$.

What about {\em multiplicativity} for manifolds glued along a surface?
Let  $M_1$ and $M_2$ be  $3$-dimensional manifolds with
$\partial M_1=\Sigma\cup \Sigma '$ and $M_2=\Sigma^*\cup \Sigma ''$, and
assume that $M_1$ is glued to $M_2$ along $\Sigma$. Then the
empty link in $M\cup M'$ is obtained as the union of the
empty link in $M$ with the empty link in $M'$.  It follows that
\begin{eqnarray*}
Z(M_1\cup M_2,m_1+m_2)=e^{-\frac{i\pi}{4}\tau}\left<Z(M_1,m_1),Z(M_2, m_2)\right>,
\end{eqnarray*}
where $\left<,\right>$ is the contraction
\begin{eqnarray*}
V(\Sigma')\otimes V(\Sigma)\otimes V(\Sigma)^*\otimes V(\Sigma'')\rightarrow
V(\Sigma')\otimes V(\Sigma''),
\end{eqnarray*}
and $\tau$ expresses the anomaly of the TQFT and depends on how
the signature of the surgery link changes under gluing (or equivalently,
on how the signatures of the $4$-dimensional manifolds bounded by the given
$3$-dimensional manifolds change under the gluing, see Section~\ref{sec:8}).

Finally, $Z(\emptyset)={\mathbb C}$, because the only link in the
void manifold is the empty link.
Also, if $M=\Sigma\times [0,1]$, then the empty link in $M$ is the
surgery diagram of the identity homeomorphism of $\Sigma$
and hence $Z(M,0)$ can be viewed as the identity map
in $\mbox{End}(V(\Sigma))$.

This TQFT is {\em hermitian} because $V(\partial M)$, being
a space of theta functions, has the inner product introduced
in Section~\ref{sec:2}. And if $M$ is a $3$-dimensional
manifold and $M^*$ is the same manifold but with reversed orientation,
then the surgery link $L^*$ of $M^*$ is the mirror image of
the surgery link $L$ of $M$. The invariant of $M$ is computed by
smoothing the crossings in $L$ while the invariant of $M^*$ is
computed by smoothing the crossings in $L^*$, whatever was a positive
crossing in $L$ becomes a negative crossing in $L^*$ and vice-versa.
Hence $\mbox{sign}(L^*)=-\mbox{sign}(L)$. Also, because
for $t=e^{\frac{i\pi}{N}}$ one has $t^{-1}=\bar{t}$, and hence
$\Omega(L^*)=\overline{\Omega(L)}$. It follows that
\begin{eqnarray*}
Z(M^*,-m)=\overline{Z(M,m)}
\end{eqnarray*}
as desired.

\section{The Hermite-Jacobi action and the non-additivity of the
signature of 4-dimensional  manifolds}\label{sec:8}

An interesting coincidence in mathematics is the fact that the
Segal-Shale-Weil cocycle of the metaplectic representation
\cite{lionvergne} and the
non-additivity of the signature of 4-dimensional
manifolds under gluings \cite{wall}
are both described in terms of the Maslov index. We explain this
coincidence by showing how to resolve the projectivity of the
Hermite-Jacobi action using 4-dimensional manifolds. Note that
the theory of theta functions explains the coincidence of the  cocycle of the
metaplectic representation and the cocycle of the Hermite-Jacobi action.

Each element $h$ of the mapping class group of $\Sigma_g$ can be
represented by surgery on a link $L_h\in \Sigma_g\times [0,1]$, meaning
that  surgery along $L_h$ yields a manifold that   is homeomorphic to
$\Sigma_g \times[0,1]$ is such a way that the homeomorphism is the identity map
on $\Sigma_g\times \{0\}$ and $h$ on $\Sigma_g\times\{1\}$.
The  link $L_h$ is not necessarily obtained from writing $h$ as a composition
of Dehn twists, as in Proposition~\ref{skeinfourier}.

\begin{theorem}
Let $h$ be an element of the mapping class group of $\Sigma_g$ obtained
by surgery on the framed link $L_h$ in $\Sigma_g\times [0,1]$.
Then  the discrete Fourier transform $\rho(h):{\mathcal L}_N(H_g)
\rightarrow {\mathcal L}_N(H_g)$ is given by
\begin{eqnarray*}
\rho(h)\beta=\Omega(L_h)\beta.
\end{eqnarray*}
\end{theorem}

\begin{proof}
Write $h=T_1T_2\cdots T_n$, where $T_j$ are Dehn twists obtained as
surgeries along the curves $\gamma_j$, $j=1,2,\ldots ,n$.
If we glue two handlebodies so as to obtain $S^3$, the gluing
defines a nondegenerate pairing $[\cdot ,\cdot ]
:{\mathcal L}_N(H_g)\times
{\mathcal L}_N(H_g)\rightarrow {\mathbb C}$. If we consider
basis elements $e_j,e_k$ in $\widetilde{\mathcal L}_t(H_g)$,
then $[\Omega(L_h)e_j,e_k]$ completely determines the operator
defined by $\Omega(L_h)$. Because of  invariance under handle slides,
we can transform $\Omega(L_h)$ into $\Omega(\gamma_1)\Omega(\gamma_2)
\cdots \Omega(\gamma_n)$ such that
\begin{eqnarray*}
[\Omega(L_h)e_j,e_k]=[\Omega(\gamma_1)\Omega(\gamma_2)\cdots \Omega(\gamma_n)e_j,e_k].
\end{eqnarray*}
Hence the operators defined by $\Omega(L_h)$ and $\Omega(\gamma_1)\Omega(\gamma_2)
\cdots \Omega(\gamma_n)$  are equal. The conclusion follows from Proposition~\ref{skeinfourier}.
\end{proof}

Fix
a Lagrangian subspace ${\bf L}$ of $H_1(\Sigma_g, {\mathbb R})$
and consider the  closed 3-manifold $M$ obtained
by gluing to the surgery of $\Sigma_g\times [0,1]$ along $L$ the
handlebodies $H_g^0$ and $H_g^1$ such that $\partial H_g^0=\Sigma_g\times \{0\}$,
$\partial H_g^1=\Sigma_g\times \{1\}$, and ${\bf L}$ respectively
$h_*({\bf L})$ are the kernels of the inclusion of $\Sigma_g$
into $H_g^0$ respectively $H_g^1$.
$M$ is the boundary
of a $4$-manifold $W$ obtained by adding $2$-handles
to the ball $B^4$ as prescribed by $L$.

The discrete Fourier transform $\rho(h)$ is a skein in
${\mathcal L}_N(\Sigma_g\times [0,1])$ which is uniquely
determined once we fix the signature of $W$.
Hence to the pair $(h,n)$ where $h$ is an element of the mapping
class group and $n\in {\mathbb Z}$, we  associate uniquely
a skein ${\mathcal F}(h,n)$, its discrete Fourier transform. We identify
the pair $(h,n)$ with $(h,\sign(W))$ where $W$ is a $4$-dimensional
manifold defined as above. Note that by adding trivial 2-handles
we can enforce $\sign(W)$ to be any integer.

Consider the ${\mathbb Z}$-extension of the mapping class
group defined by the multiplication rule
\begin{eqnarray*}
(h',\sign(W'))\circ(h,\sign(W))=(h'\circ h, \sign(W'\cup W))
\end{eqnarray*}
where $W'$ and $W$ are glued in such a way that $H_g^0\in W'$ is
identified with $H_g^1$ in $W$. If ${\bf L}'$ is the Lagrangian
subspace of $H_1(\Sigma_g,{\mathbb R})$
 used for defining $W'$, then necessarily ${\bf L}'=h_*({\bf L})$.
Recall Wall's formula for the non-additivity of the signature
of 4-dimensional manifolds
\begin{eqnarray*}
\sign(W'\cup W)=\sign(W')+\sign(W)-\tau({\bf L},h_*({\bf L}), h'_*\circ
h_*({\bf L})),
\end{eqnarray*}
where $\tau $ is the Maslov index.
By using this formula we obtain
\begin{alignat*}{1}
{\mathcal F}&(h'\circ h,\sign(W'\cup W))\\
&={\mathcal F}(h'\circ h,
\sign(W')+\sign(W)-\tau({\bf L}, h_*({\bf L}), h'_*\circ h_*({\bf L}))\\
&=e^{-\frac{i\pi}{4}\tau({\bf L},h_*({\bf L}),h'_*\circ h_*({\bf L}))}{\mathcal
  F}(h'\circ h, \sign(W')+\sign(W))\\
&=e^{-\frac{i\pi}{4}\tau({\bf L},h_*({\bf L}),h'_*\circ h_*({\bf L}))}{\mathcal
  F}(h',\sign(W'))
{\mathcal F}(h,\sign(W)),
\end{alignat*}
where for the second step we changed the signature of the 4-dimensional
manifold associated to $h\circ h'$ by adding trivial handles and
used Proposition~\ref{omegatrivial}.

Or equivalently, if we let $\rho(h,\sign(W))$ be discrete
the Fourier transform associated to $h$, normalized by the (signature of)
the manifold $W$, then
\begin{eqnarray*}
\rho(h'\circ h,\sign(W'\cup W))=e^{-\frac{i\pi}{4}\tau({\bf L},h_*({\bf
    L}),h'_*\circ h_*({\bf L}))}
\rho(h',\sign(W'))\rho(h,\sign(W)).
\end{eqnarray*}
In this formula we recognize the  Segal-Shale-Weil cocycle
\begin{eqnarray*}
c(h,h')=e^{-\frac{i\pi}{4}\tau({\bf L},h_*({\bf L}),h'_*\circ h_*({\bf L}))}
\end{eqnarray*}
used for resolving the projective ambiguity of the Hermite-Jacobi action,
as well as for resolving the projective ambiguity of the metaplectic
representation.

\section{Theta functions and abelian Chern-Simons theory}\label{sec:9}

By Jacobi's inversion theorem and Abel's theorem \cite{farkaskra},
the Jacobian of a surface $\Sigma_g$
pa\-ram\-etrizes the set of divisors of degree zero
modulo principal divisors. This is the moduli space of
stable line bundles, which is the same as the moduli space
${\mathcal M}_g(U(1))$
of flat $u(1)$-connections on the surface (in the trivial $U(1)$-bundle).

The moduli space ${\mathcal M}_g(U(1))$
has a complex structure defined as follows (see for
example \cite{hitchin}). The tangent space to ${\mathcal M}_g(U(1))$
at an arbitrary point is $H^{1}(\Sigma_g,{\mathbb R})$, which, by Hodge
theory, can be identified with
the space of real-valued harmonic $1$-forms on $\Sigma_g$. The complex
structure is given by
\begin{eqnarray*}
J\alpha=-*\alpha,
\end{eqnarray*}
where $\alpha$ is a harmonic form. In local coordinates, if
$\alpha=udx+vdy$, then $J(udx+vdy)=vdx-udy$.

If we identify the space of real-valued harmonic $1$-forms with the
space of holomorphic $1$-forms $H^{(1,0)}(\Sigma_g)$ by the map
$\Phi$ given in local coordinates by $\Phi(udx+vdy)=(u-iv)dz$, then
the complex structure becomes multiplication by $i$ in $H^{(1,0)}(\Sigma_g)$.

The moduli space is a torus obtained
by exponentiation
\begin{eqnarray*}
{\mathcal M}_{U(1)}=H^1(\Sigma_g,{\mathbb R})/{\mathbb Z}^{2g}.
\end{eqnarray*}
If we choose a basis of the space of real-valued
harmonic forms $\alpha_1,\alpha_2,\ldots, \alpha_g,$ $\beta_1,
\ldots, \beta_g$ such that
\begin{eqnarray}\label{rightperiods}
\int_{a_j}\alpha_k=\delta_{jk}, \quad \int_{b_j}\alpha_k=0, \quad
\int_{a_j}\beta_k=0,\quad \int_{b_j}\beta_k=\delta_{jk},
\end{eqnarray}
then the above ${\mathbb Z}^{2g}$ is the period matrix of this basis.

On the other hand, if $\zeta_1,\zeta_2,\ldots, \zeta_g$ are
the holomorphic forms introduced in Section~\ref{sec:2}, and
if $\alpha_j'=\Phi^{-1}(\zeta_j)$ and $\beta_{j}'=\Phi^{-1}(-i\zeta_j)$,
$j=1,2,\ldots, g$ then one can compute that
 \begin{eqnarray*}
\int_{a_j}\alpha_k'=\delta_{jk},\, \int_{b_j}\alpha_k'=\mbox{Re}\ \pi _{jk},\,
\int_{a_j}\beta_k'=0, \, \int_{b_j}\beta_k'=\mbox{Im} \ \pi_{jk}.
\end{eqnarray*}
The basis $\alpha_1',\ldots, \alpha_g',\beta_1',\ldots, \beta_g'$ determines
coordinates $(X',Y')$ in the tangent space to ${\mathcal M}_{U(1)}$.
If we consider the change of coordinates
$X'+iY'=X+\Pi Y$,  then the moduli space is the quotient of
${\mathbb C}^g$ by the integer lattice ${\mathbb Z}^{2g}$.
This is exactly what has been
done in Section~\ref{sec:2} to obtain the   jacobian variety.
This shows that the complex structure on the Jacobian variety coincides with
 the standard complex structure on the moduli space of flat $u(1)$-connections
on the surface.

The moduli space ${\mathcal M}_g(\Sigma_g)$ has a symplectic
structure defined by the Atiyah-Bott form \cite{atiyahbott}. This
form is  given by
\begin{eqnarray*}
\omega(\alpha,\beta)=-\int_{\Sigma_g}\alpha\wedge \beta,
\end{eqnarray*}
where $\alpha,\beta$ are real valued harmonic $1$-forms,
i.e. vectors in the tangent
space to ${\mathcal M}_g(\Sigma_g)$. If $\alpha_j,\beta_j$, $j=1,2,\ldots, g$
are as in (\ref{rightperiods}), then $\omega(\alpha_j,\alpha_k)=
\omega(\beta_j,\beta_j)=0$ and $\omega(\alpha_j,\beta_k)=\delta_{jk}$
(which can be seen by identifying the space of real-valued harmonic $1$-forms
 with $H^1(\Sigma_g,{\mathbb R})$ and using the topological definition of
the cup product). This shows that the Atiyah-Bott
form coincides with the symplectic form introduced in Section~\ref{sec:2}.

For a $u(1)$-connection $A$ and curve $\gamma$ on the surface,
we denote by $\mbox{hol}_\gamma (A)$ the holonomy of $A$ along $\gamma$.
The map $A\mapsto \mbox{trace}(\mbox{hol}_\gamma(A))$
 induces a function
 on the Jacobian  variety called Wilson line\footnote{Since we work with
the group $U(1)$,
the holonomy is just a complex number of absolute value $1$, and the
trace is the number itself.}. If
$[\gamma]=(p,q)\in H_1(\Sigma_g,{\mathbb Z})$,
then the Wilson line associated to $\gamma$ is the function
$(x,y)\mapsto \exp 2\pi i(p^Tx+q^Ty)$. These are the functions on
the Jacobian variety of interest to us.

The goal is to quantize the moduli space of flat $u(1)$-connections
on the closed Riemann surface $\Sigma_g$ endowed with the
Atiyah-Bott symplectic form. One procedure has been
outlined in Section~\ref{sec:2}; it is Weyl quantization on the
$2g$-dimensional torus in the holomorphic polarization.

Another quantization procedure has been introduced by Witten
in \cite{witten} using Feynman path integrals.
In his approach, states and observables are defined by path integrals
of the form
\begin{eqnarray*}
\int_{A}e^{\frac{i}{\hh}{
    L}(A)}
\mbox{trace}(\mbox{hol}_{\gamma}(A)){\mathcal D}A,
\end{eqnarray*}
where $L(A)$ is the Chern-Simons lagrangian
\begin{eqnarray*}
{ L}(A)=\frac{1}{4\pi}\int_{\Sigma_g\times [0,1]}
\mbox{tr}\left(A\wedge dA+\frac{2}{3}A\wedge A\wedge A\right).
\end{eqnarray*}
 According to
Witten, states and observables should be representable as
skeins in the skein modules of the linking number
 discussed in Section~\ref{sec:4}.

Witten's quantization model is symmetric with respect to the action
of the mapping class group of the surface, a property shared by
Weyl quantization in the guise of the exact Egorov identity
(\ref{egorov}). As we have  seen, the two quantization
models coincide.

It was Andersen \cite{andersen}
 who pointed out that the quantization of the Jacobian
that arises in Chern-Simons theory coincides with Weyl quantization.
For non-abelian Chern-Simons theory, this phenomenon was first observed
by the authors in \cite{gelcauribe1}.

In the sequel of this paper, \cite{gelcauribe}, we will conclude that,
for non-abelian Chern-Simons, the algebra of quantum group quantizations
of Wilson lines on a surface
and the Reshetikhin-Turaev representation of the mapping
class group of the surface are analogues of the group algebra of the
finite Heisenberg group and of the Hermite-Jacobi action. We will show
 how the element $\Omega$ corresponding to  the group $SU(2)$ can be derived
by studying the discrete sine transform.

\bibliographystyle{amsplain}

\end{document}